\documentclass[a4paper,11pt]{article}
\usepackage[margin=1in]{geometry}
\usepackage[all=normal,bibliography=tight]{savetrees}
\usepackage{microtype}

\usepackage{amsthm}
\usepackage{amsmath}
\usepackage{amssymb}

\usepackage{comment}
\usepackage{todonotes}
\usepackage{xspace}

\usepackage[shortlabels,inline]{enumitem}
\setlist{noitemsep}
\setlist[enumerate,1]{label=\arabic*., ref=(\arabic*)}

\usepackage{float}
\usepackage{graphics}
\usepackage{xcolor}
\definecolor{dark-red}{rgb}{0.4,0.15,0.15}
\definecolor{dark-blue}{rgb}{0.15,0.15,0.4}
\definecolor{medium-blue}{rgb}{0,0,0.5}
\definecolor{gray}{rgb}{0.5,0.5,0.5}

\usepackage{url}
\usepackage[pdftex]{hyperref}
\def\cqedsymbol{\ifmmode$\lrcorner$\else{\unskip\nobreak\hfil
\penalty50\hskip1em\null\nobreak\hfil$\lrcorner$
\parfillskip=0pt\finalhyphendemerits=0\endgraf}\fi}

\let\plainqed\qedsymbol
\newcommand{\claimqed}{$\lrcorner$}

\newenvironment{claimproof}{\begin{proof}\renewcommand{\qedsymbol}{\claimqed}}{\end{proof}\renewcommand{\qedsymbol}{\plainqed}}

\newcommand{\V}[0]{\ensuremath{\mathcal{V}}\xspace}
\newcommand{\W}[0]{\ensuremath{\mathcal{W}}\xspace}
\newcommand{\CC}[0]{\ensuremath{\mathcal{C}}\xspace}
\newcommand{\X}[0]{\ensuremath{\mathcal{X}}\xspace}
\newcommand{\Y}[0]{\ensuremath{\mathcal{Y}}\xspace}
\newcommand{\Z}[0]{\ensuremath{\mathcal{Z}}\xspace}
\newcommand{\T}[0]{\ensuremath{\mathcal{T}}\xspace}

\newcommand{\torso}[0]{\mathop{\mathrm{\textsc{torso}}}}

\newcommand{\Oh}[0]{\ensuremath{\mathcal{O}}\xspace}

\newlength{\baseImageHeight}
\newcommand{\hyphen}{\nobreakdash-\hspace{0pt}}
\newcommand{\kPath}[0]{\textsc{$k$\hyphen Path}\xspace}

\newcommand{\kCycle}[0]{\textsc{$k$\hyphen Cycle}\xspace}

\newcommand{\kLinkage}[0]{\textsc{Auxiliary Linkage}}

\newtheorem{theorem}{Theorem}[section]

\newtheorem{corollary}[theorem]{Corollary}
\newtheorem{lemma}[theorem]{Lemma}
\newtheorem{numberedclaim}[theorem]{Claim}

\theoremstyle{definition}

\newtheorem{definition}[theorem]{Definition}

\hypersetup{
    colorlinks, linkcolor={dark-red},
    citecolor={dark-blue}, urlcolor={medium-blue}
}

\date{}

\begin{document}
\title{Turing Kernelization for Finding Long Paths \\in Graph Classes Excluding a Topological Minor\thanks{
This work was supported by the Netherlands Organization for Scientific Research (NWO) Veni grant 639.021.437 ``Frontiers in Parameterized Preprocessing'' and Gravitation grant 024.002.003 ``Networks''.
Marcin Pilipczuk is supported by the ``Recent trends in kernelization: theory and experimental evaluation'' project, carried out within the Homing programme of the Foundation for Polish Science co-financed by the European Union under the European Regional Development Fund.
Marcin Wrochna is supported by the National Science Centre of Poland grant number 2013/11/D/ST6/03073 and by the Foundation for Polish Science (FNP) via the START stipend programme.
}}
\author{Bart M.\ P.\ Jansen\thanks{Eindhoven University of Technology, The Netherlands}, Marcin Pilipczuk\thanks{University of Warsaw, Poland}, and Marcin Wrochna\footnotemark[3]}

\maketitle

\begin{abstract}
The notion of Turing kernelization investigates whether a polynomial-time algorithm can solve an NP-hard problem, when it is aided by an oracle that can be queried for the answers to bounded-size subproblems. One of the main open problems in this direction is whether \kPath admits a polynomial Turing kernel: can a polynomial-time algorithm determine whether an undirected graph has a simple path of length~$k$, using an oracle that answers queries of size~$k^{\Oh(1)}$?

We show this can be done when the input graph avoids a fixed graph~$H$ as a topological minor, thereby significantly generalizing an earlier result for bounded-degree and~$K_{3,t}$-minor-free graphs. Moreover, we show that \kPath even admits a polynomial Turing kernel when the input graph is not $H$-topological-minor-free itself, but contains a known vertex modulator of size bounded polynomially in the parameter, whose deletion makes it so. To obtain our results, we build on the graph minors decomposition to show that any $H$-topological-minor-free graph that does not contain a $k$-path, has a separation that can safely be reduced after communication with the oracle.
\end{abstract}

\section{Introduction}
Suppose that Alice is a polynomial-time agent faced with an input to an NP-hard problem that she wishes to solve exactly. To facilitate her in this process, she can ask questions to an all-knowing oracle. These will be answered truthfully and instantly, but the oracle is memory-less and will not take previous questions into account when answering the next one. How large do these questions have to be, to allow Alice to find the answer to her problem? Clearly, the answer can be established by sending the entire input to the oracle, who determines the answer and sends it to Alice. Could there be a more clever strategy? Alice can attempt to isolate a small but meaningful question about the behavior of her input, such that after learning its answer, she can reduce to a smaller input without changing the outcome. Iterating this process solves her problem: when it has become sufficiently small, it can be posed to the oracle in its entirety.

Such problem-solving strategies can be rigorously analyzed using the notion of \emph{Turing kernelization} that originated in parameterized algorithmics. The parameter makes it possible to express how the size of the questions that Alice asks, depends on properties of the input that she is given. (See Section~\ref{sec:turing} for a formal definition.)

Understanding the power of Turing kernelization is one of the main open research horizons in parameterized algorithmics. There is a handful of problems for which a nontrivial Turing kernelization is known~\cite{AmbalathBHKMPR10,BarberoPP17,Binkele-RaibleFFLSV12,BodlaenderJK14,GarneroW16,HuffnerKS15,JansenM15,KolayP15,SchaferKMN12,ThomasseTV14}. On the other hand, there is a hierarchy of parameterized complexity classes which are conjectured not to admit polynomial Turing kernels~\cite{HermelinKSWW15}. Arguably, the main open problem (cf.~\cite{BodlaenderDFGHLMRRR08,Binkele-RaibleFFLSV12,HermelinKSWW15}) in this direction is to determine whether the \kPath problem (determine whether an undirected graph has a simple path of length~$k$) has a polynomial Turing kernel. In earlier work~\cite{Jansen17}, the first author showed that \kPath indeed admits polynomial Turing kernels on several graph classes. In this work, we develop Turing kernels for \kPath in a much more general setting.

\subparagraph{Our results}

Our algorithmic contributions are twofold. First of all, we extend the Turing kernelization for \kPath to much broader families of sparse graphs. Whereas the earlier work could only deal with $K_{3,t}$-minor-free graphs, claw-free graphs, and bounded-degree graphs, we show that a Turing kernelization exists on $H$-minor-free graphs for all fixed graphs~$H$. We even lift the kernelization to $H$-topological-minor-free graphs, thereby capturing a common generalization of the bounded-degree and~$K_{3,t}$-minor-free cases.

\begin{theorem}\label{thm:topo-minor-main}
For every fixed graph $H$, the \kPath problem, restricted to graphs excluding $H$ as a topological minor,
admits a polynomial Turing kernel.
Furthermore, the kernel runs in time $k^{\Oh_H(1)} n^2 m$ and invokes $k^{\Oh_H(1)} \cdot n$ calls
to the oracle.
\end{theorem}

Our second contribution is the following theorem. By a novel algorithmic approach, we obtain a Turing kernelization even when the input graph does not belong to the desired restricted graph class itself, but contains a small known vertex modulator whose deletion places the graph in such a graph class.

\begin{theorem}\label{thm:minor-mod-main}
For every fixed graph $H$, the \kPath problem, on instances consisting of a graph $G$, integer $k$, and a modulator $M \subseteq V(G)$ such that $G-M$ is $H$-topological-minor-free,
admits a polynomial Turing kernel, when parameterized by $k$ and $|M|$.
\end{theorem}

\subparagraph{Techniques}

To explain our approach, we briefly recall the idea behind the Turing kernelization for \kPath on planar graphs. At the core lies a win/win: there is a polynomial-time algorithm that either (i) establishes that a planar graph~$G$ has a $k$-path (a simple path on $k$ vertices), or (ii) finds a separation~$(A,B)$ in~$G$ with the following property: the size of~$A$ is polynomially bounded in~$k$, but large enough that after marking a witness structure for each reasonable way in which a $k$-path might intersect~$A$, some vertex remains unmarked. Using bounded-size oracle queries to mark the witness structures, this allows the problem to be simplified by removing an unmarked vertex from~$A$ without changing the answer.

Theorem~\ref{thm:topo-minor-main} is established by lifting this win/win approach to $H$-(topological)-minor-free graphs. This requires an adaptation of the decomposition theorems of Robertson and Seymour~\cite{RobertsonS99} (for minors) and of Grohe and Marx~\cite{GroheM15} (for topological minors), to obtain the following. Every $H$-free graph that does not have a $k$-path, has a tree decomposition of constant adhesion and width~$\mathrm{poly}(k)$. A reducible separation can be found by inspecting this tree decomposition. To establish this result, we exploit known theorems stating that triconnected $n$-vertex graphs that exclude~$K_{3,t}$ as a minor for some~$t$~\cite{ChenYZ12}, contain paths of length~$\Omega(n^{\varepsilon})$ for some~$\varepsilon > 0$. Roughly speaking, this allows us to infer the existence of a $k$-path if there is a large embedded part in the nearly-embeddable graph corresponding to a bag of the graph minors decomposition, since graphs embeddable in a fixed surface are~$K_{3,t}$-minor-free for some~$t$. We use lower bounds on the circumference of graphs of bounded degree~\cite{ChenGYZ06,Shan15} to achieve a similar conclusion from the existence of a large bounded-degree bag in the topological-minor-free decomposition. Several technical steps are needed to translate this into the desired win/win, due to the existence of vortices, virtual edges, and the lack of a direct polynomial-time algorithm to compute the decomposition.

To prove Theorem~\ref{thm:minor-mod-main}, we introduce a new algorithmic tool for finding irrelevant vertices for the \kPath problem in the presence of a modulator~$M$ in the input graph~$G$. Since Theorem~\ref{thm:topo-minor-main} can be applied to find a $k$-path in~$G - M$ if one exists, the challenge is to detect a $k$-path in~$G$ that jumps between~$M$ and~$G-M$ several times. The absence of a $k$-path in~$G-M$ implies it has a tree decomposition of width~$\textrm{poly}(k)$ and constant adhesion. Using Theorem~\ref{thm:topo-minor-main} as a subroutine, along with a packing argument, we can compute a vertex set~$X$ of size polynomial in~$k + |M|$ with the following guarantee. If there is a $k$-path, then there is a \emph{guarded} $k$-path~$P$ in which each successive pair of vertices in~$M \cap P$ are connected by a subpath through~$G - M$ that intersects~$X$. Using the tree decomposition of~$G-M$, the standard ancestor-marking technique allows us to identify a vertex subset~$C$ of~$G-(M\cup X)$ that is adjacent to constantly many vertices from~$X$. Unless~$G$ is already small, we can find such a set~$C$ that is sufficiently large to be reducible but small enough that we may invoke the oracle for questions about it. We can then reduce the graph without losing the existence of a guarded $k$-path, by marking a witness for each sensible way in which a constant-size subset from~$M$ can connect to prescribed vertices in~$X$ through~$C$. The fact that~$C$ only has constantly many neighbors in~$X$ implies that there are only polynomially many relevant choices. We may then safely remove the unmarked vertices.

\subparagraph{Organization}

After preliminaries in Section~\ref{sec:prelims}, we give a generic Turing-style reduction rule for \kPath in Section~\ref{sec:turing}. In Section~\ref{sec:minors} we show that an $H$-minor-free graph either has a $k$-path or a separation that is suitable for reduction. In Section~\ref{sec:topological} we extend this to topological minors. Finally, in Section~\ref{sec:modulator} we present a Turing kernel applicable when the input graph has a small modulator to a suitable graph class.

\section{Preliminaries}\label{sec:prelims}
\paragraph{Notation.}
All graphs we consider are finite, simple, and undirected. A \emph{separation} of a graph $G$ is a pair $(A,B)$, $A,B \subseteq V(G)$ such that $A \cup B = V(G)$ and there are no edges between $A \setminus B$ and $B \setminus A$. The \emph{order} of the separation $(A,B)$ is $|A \cap B|$. A graph is \emph{triconnected} if it is connected and cannot be disconnected by deleting fewer than three vertices. When referring to the \emph{size} of a graph in our statements, we mean the number of vertices.

A \emph{tree decomposition} of a graph~$G$ is a pair~$(T,\X)$ where~$T$ is a rooted tree and~$\X$ is a function that assigns to every node~$t \in V(T)$ a subset~$\X(t)$ of $V(G)$ called a \emph{bag} such that:
\begin{itemize}
\item $\bigcup_{t\in V(T)} \X(t) = V(G)$;
\item for each edge $uv \in E(G)$, there is a node $t \in V(T)$ with $u, v \in \X(t)$;
\item for each $v \in V(G)$, the nodes $\{t \mid v \in \X(t)\}$ induce a (connected) subtree of $T$.
\end{itemize}
The \emph{width} of $(T,\X)$ is $\max_{t \in V(T)} |\X(t)|-1$. 
Its \emph{adhesion} is $\max_{tt'\in E(T)} |\X(t)\cap\X(t')|$.
We also call the set $\X(t) \cap \X(t')$ the \emph{adhesion of $tt'$}, for every edge $tt'$ of $T$.
%If T has no edges, we define the adhesion to be zero.
%For an edge e = \{i, j\} ∈ E(T ) we will sometimes refer to the set X (i) ∩ X (j) as the adhesion of edge e.
For a decomposition $(T,\X)$ of $G$ and a node $t \in V(T)$, the \emph{torso}, denoted $\torso(G,\X(t))$, is the graph obtained from $G[\X(t)]$ by adding an edge between each pair of vertices in $\X(t) \cap \X(t')$, for every neighbor $t'$ of $t$ in $T$ (so each adhesion induces a clique in the torso). %that are connected by a path in $G$ whose internal vertices are outside of $\X(t)$.
Added edges not present in $G$ are called \emph{virtual edges}.
For a subtree $T'\subseteq T$ we write $\X(T')$ for the union $\bigcup_{t\in V(T')} \X(t)$ of bags in $T'$.

For an edge $t_1t_2 \in E(T)$, let $T_i$ be the connected component of $T-\{t_1t_2\}$ that
contains $t_i$. Let $V_i = \X(T_i)$. Observe that the properties of a tree decomposition
imply that $(V_1,V_2)$ is a separation with $V_1 \cap V_2 = \X(t_1) \cap X(t_2)$.

A decomposition $(T,\X)$ is \emph{connected} if for every $t \in V(T)$ and its child $t'$, if $T_{t'}$ is the subtree of $T$ rooted at $t'$, we have (i)~that $G[\X(T_{t'}) \setminus \X(t)]$ is connected, and (ii)~that~$\X(T_{t'}) \setminus \X(t)$ has edges to every vertex of the adhesion~$\X(t) \cap \X(t')$.
It is straightforward to turn any decomposition into a connected one without increasing its width nor adhesion. For~(i), as long as there exists a pair $(t,t')$
violating the requirement, make a distinct copy $T_{t'}^C$ of $T_{t'}$ for each connected component $C$ of $G[\X(T_{t'}) \setminus \X(t)]$, restrict
the bags of $T_{t'}^C$ to the vertices of $C \cup \X(t)$ only, and attach $T_{t'}^C$ as a subtree with the root being a child of $t$. For~(ii), while there is a vertex~$v \in \X(t) \cap \X(t')$ that has no neighbors in~$\X(T_{t'}) \setminus \X(t)$, simply remove~$v$ from all bags in~$T_{t'}$.

We will also need the following non-standard complexity measure of a tree decomposition
$(T,\X)$. For every $t \in V(T)$, the number of distinct adhesions $\X(t) \cap \X(t')$ for
$t' \in N_T(t)$ is called the \emph{adhesion degree} of $t$. 
The maximum adhesion degree over all nodes $t$ is the \emph{adhesion degree} of the decomposition
$(T,\X)$.
Observe that if a tree decomposition $(T,\X)$ has width less than $\ell$ and adhesions
of size at most $h$, then its
adhesion degree is at most 
$$\sum_{i = 0}^h \binom{\ell}{i} \leq (1+\ell)^h.$$
However, in sparse graph classes we can prove a much better
bound on the adhesion degree due to linear bounds on the number of cliques in such graphs; cf. Lemma~\ref{lem:adh-deg-bound}.

A \emph{path decomposition} is a tree decomposition where $T$ is a path; we will denote it simply as a sequence of bags $\Z_1,\dots,\Z_{|V(T)|}$.

For an integer $\theta$, a \emph{tangle of order $\theta$} in a graph $G$
is a family $\T$ of separations of order at most $\theta$ such that, for every separation
$(X,Y)$ of order at most $\theta$ in $G$, exactly one of the separations $(X,Y)$ and $(Y,X)$
belongs to $\T$. If $(X,Y) \in \T$ we call $X$ the \emph{small side}
of the separation and $Y$ the \emph{large side}. 
Furthermore, we require that for every three separations $(X_1,Y_1)$, $(X_2, Y_2)$, $(X_3,Y_3)$ in $\T$, we have $G[X_1] \cup G[X_2] \cup G[X_3] \neq G$.

%For a tangle $\T$ of order $\theta$ and a set $A \subseteq V(G)$ of size $|A| \leq \theta$,
%we define a tangle $\T-A$ of order $\theta-|A|$ in $G-A$ as follows: for every
%separation $(X,Y)$ of order at most $\theta-|A|$ in $G-A$, we put
%$(X,Y)$ into $\T-A$ if $(X \cup A, Y \cup A) \in \T$, and otherwise,
%if $(Y \cup A, X \cup A) \in \T$, we put $(Y,X)$ into $\T-A$.

\section{Turing kernels}\label{sec:turing}
\newcommand{\sepmin}{p}
\newcommand{\sepmax}{q}
\newcommand{\sepord}{h}
\newcommand{\tdw}{w}
\newcommand{\tdadh}{h}
\newcommand{\tddeg}{a}
\newcommand{\Ibnd}{\mathfrak{p}}
\newcommand{\GG}{\mathcal{G}}

In this section we introduce a general toolbox and notation for proving our Turing kernel
bounds.

\subsection{Definitions and the auxiliary problem}

For a parameterized problem $\Pi$ and a computable function $f$, a \emph{Turing kernel of size $f$}
is an algorithm that
solves an input instance $(x,k)$ of $\Pi$ in polynomial time, given access 
to an oracle that solves instances $(x',k')$ of $\Pi$ with $|x'|,k' \leq f(k)$. 
A Turing kernel is a \emph{polynomial} one if $f$ is a polynomial.

If we are only interested in distinguishing between NP-complete
problems admitting a polynomial Turing kernel
from the ones that do not admit such a kernel, we can assume that the oracle solves
an arbitrary problem in NP, not necessarily the \kPath{} problem. 
Indeed, note that by the definition of NP-completeness, an oracle to a problem in NP
can be implemented with an oracle to \kPath{} with only polynomial blow-up in the size 
of the passed instances.

In our work, it will be convenient to reduce to the \kLinkage{} problem, defined as follows.
The input consists of an undirected graph $G'$, an integer $k'$, a set of terminals $S \subseteq V(G')$, and a number of requests $R_1,R_2,\ldots,R_r$; a request is a set of at most two terminals.
A path $P_i$ in $G$ is said to \emph{satisfy a request $R_i$} if $V(P_i) \cap S = R_i$ and every
vertex of $V(P_i) \cap S$ is an endpoint of $P_i$. 
With such an input, the \kLinkage{} problem asks for a sequence of $r$ paths $P_1,P_2,\ldots,P_r$
such that $P_i$ satisfies $R_i$ for every $1 \leq i \leq r$, $|\bigcup_{i=1}^r V(P_i)| = k'$,
 and every vertex of $V(G)\setminus S$ is contained in at most
one path $P_i$ (i.e., the paths $P_i$ are vertex-disjoint, except that they may share an endpoint, but only if the requests ask them to do so). 

We remark that \kLinkage{} is a more general problem than \kPath{}: an instance
with $G'=G$, $k'=k$, $S=\emptyset$, $r=1$, and $R_1 = \emptyset$ asks precisely for a $k$-path
in $G$.

Clearly, the decision version of the \kLinkage{} problem belongs to the class NP.
By using its self-reducibility (cf.~\cite[Lemma 2]{Jansen17}), we assume that the oracle returns a sequence of paths $(P_i)_{i=1}^r$ in case of a positive answer.
That is, in all subsequent bounds on the number of \kLinkage{} oracle calls, the bound adheres to the number of calls to an oracle that returns the actual paths $P_i$;
if one wants to use a decision oracle, one should increase the bound by the blow-up implied by the self-reducibility application (i.e., at most $|E(H)|$ for calls on a graph $H$). 

\subsection{Generic reduction rule}

\begin{figure}
	\centering
    \def\svgwidth{200pt}
	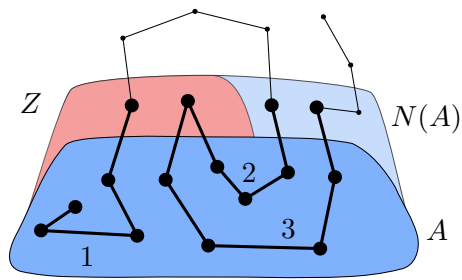
	\vspace*{-1.5em}
	\caption{A set $A$ (blue) and a path with three $A$-traverses (bold). The path is guarded w.r.t.~$Z \subseteq N(A)$ (the red set), since each $A$-traverse has an endpoint in it. (A path fully contained in $A$ (with one $A$-traverse) or disjoint from $A$ (with no $A$-traverses) would also be guarded.)}
	\label{fig:guardedPath}
\end{figure}

We now show a generic reduction rule for the \kPath{} problem.
We start with a few definitions.

\begin{definition}
For a graph $G$, a subset $A \subseteq V(G)$, and a simple path $P$ in $G$, 
an \emph{$A$-traverse} of $P$ is 
a maximal subpath of $P$ that contains at least one vertex of $A$ and 
has all its internal vertices in $A$.
\end{definition}

Note that if $Q$ is an $A$-traverse of $P$, then every endpoint of $Q$ is either
an endpoint of $P$ or lies in $N_G(A)$. See Figure~\ref{fig:guardedPath}.

\begin{definition}
Let $G$ be a graph, $A \subseteq V(G)$, and let $k$ be an integer. A set $Z \subseteq N(A)$ is called a \emph{$k$-guard of $A$}
if the following implication holds: if $G$ admits a $k$-path, then there exists a $k$-path
$P$ in $G$ that is either contained in $A$ or such that every $A$-traverse of $P$ 
has at least one endpoint in $Z$.
\end{definition}
Given a graph $G$, a set $A \subseteq V(G)$, and a $k$-guard $Z \subseteq N(A)$ of $A$, a $k$-path $P$
satisfying properties as in the above definition is called \emph{guarded} (w.r.t.~$k$, $A$, and $Z$).
If the integer $k$ and the set $A$ are clear from the context, we call such a set $Z$ simply a guard.

Observe that $Z = N(A)$ is always a guard, but sometimes we will be able to find smaller ones.
Of particular interest will be guards of constant size, as our kernel sizes will depend exponentially
on the guard size. To describe our single reduction rule, we show how solutions to \kLinkage{} can be used to preserve the existence of guarded $k$-paths.

Assume we are given a graph $G$, a set $A \subseteq V(G)$, an integer $k$, and a $k$-guard $Z \subseteq N(A)$ of $A$.
Let $h = |Z|$ and $\ell = |N(A)|$. 
Furthermore, assume that $G$ admits a $k$-path, and let $P$ be a guarded one w.r.t.~$A$ and $Z$.
Let $(Q_1,Q_2,\ldots,Q_r)$ be the $A$-traverses of $P$, let $R_i = V(Q_i) \setminus A = V(Q_i) \cap N(A)$ for $1 \leq i \leq r$, let $G' = G[N[A]]$, $S = N(A)$, and let $k' = |\bigcup_{i=1}^r V(Q_i)|$.
Observe that $(Q_1,Q_2,\ldots,Q_r)$ is a feasible solution to the \kLinkage{} instance
$\mathcal{I}_P := (G[N[A]], k', S, (R_i)_{i=1}^r)$; the instance $\mathcal{I}_P$ is henceforth called \emph{induced by $P$ and $A$}.
Furthermore, it is easy to see that if $(Q_1',Q_2',\ldots,Q_r')$ is a different
feasible solution to $\mathcal{I}_P$, then a path $P'$ obtained from $P$ by replacing every
subpath $Q_i$ with $Q_i'$ is also a guarded $k$-path in $G$.

The crucial observation is that a small guard limits the number of $A$-traverses.
\begin{lemma}\label{lem:num-traverses}
The number~$r$ of traverses of the guarded $k$-path~$P$ is bounded by~$\max(1, 2|Z|)$.
\end{lemma}
\begin{proof}
Every vertex of $Z$ can be an endpoint of at most two traverses. 
If $r > 1$, then none of the traverses $Q_i$ are contained in $G[A]$, and thus every traverse has at least
one endpoint in the guard~$Z$. 
\end{proof}
Lemma~\ref{lem:num-traverses} in turn limits the number of possible instances $\mathcal{I}$ that can be induced by a guarded $k$-path, 
for a fixed set $A$ and guard $Z$.
Note that we have $0 \leq k' \leq k$ and $0 \leq r \leq \max(1, 2|Z|)$.
Furthermore, 
unless $r=1$ and $R_1 = \emptyset$,
we have $R_i \subseteq N(A)$, $|R_i| \in \{1,2\}$, and every set $R_i$ needs to have at least one element of $Z$;
there are at most $|Z| \cdot (|N(A)|+1) = h(\ell+1)$ choices for such a set $R_i$.
Consequently, the number of possibilities for the instance $\mathcal{I}$ is at most
\begin{equation}\label{eq:Ibound}
(k+1) \cdot \left(1 + \sum_{r=0}^{2h} h^r (\ell+1)^r \right) \leq (k+1) \cdot \left(h(\ell+1)\right)^{2h+1} =: \Ibnd(k, \ell, h).
\end{equation}
\paragraph{Reduction rule.}
If $|A| > k \cdot \Ibnd(k, \ell, h)$, then we can apply the following reduction rule.
For each \kLinkage{} instance $\mathcal{I}$ out of at most $\Ibnd(k, \ell, h)$ reasonable instances
for $A$-traverses of a guarded $k$-path in $G$, we invoke an oracle on the instance $\mathcal{I}$,
and mark the vertices of the solution if the oracle finds one.
The whole process will mark at most $k \cdot \Ibnd(k, \ell, h) < |A|$ vertices, thus at least
one vertex of $|A|$ will remain unmarked.
We delete any such vertices.

The observation that on a guarded $k$-path $P$ one can replace a solution to the instance
$\mathcal{I}_P$ induced by $P$ and $A$ by a different solution provides safeness
of this reduction.
Finally, note that the reduction invokes at most $\Ibnd(k, \ell, h)$ calls to the oracle;
each call operates on a subgraph of the graph $G[N[A]]$ with $k' \leq k$ and $r \leq 2|Z|$.

\medskip

We shall apply the Reduction Rule for a medium-sized set $A$ and a guard set $Z$ of constant size formed from adhesions of a tree decomposition. For most of the paper we will use $Z = N(A)$ with $\ell=h=|Z|$ a constant (depending on the excluded (topological) minor, in the results of Sections~\ref{sec:minors} and~\ref{sec:topological}). Only in Section~\ref{sec:modulator}, when dealing with a modulator $M$ such that $G-M$ has an appropriate structure, it will be important to consider $N(A)$ potentially containing all of $M$, with a guard set $Z$ of constant size disjoint from $M$.

\subsection{Separation oracles}

The natural way of using our reduction rule is to find in a graph a large (but not too large) part of the graph with a small (preferably, constant) boundary.
Let us first make an abstract definition of an algorithm finding such a separation.

\begin{definition}\label{def:separationOracle}
For a graph class $\GG$, a constant $\sepord$, and a computable coordinate-wise nondecreasing
function $\sepmax : \mathbb{Z}_{\geq 0} \times \mathbb{Z}_{\geq 0} \to \mathbb{Z}_{\geq 0}$,
an algorithm $\mathcal{S}$ is called a \emph{$(\sepord,\sepmax,T_{\mathcal{S}})$-separation oracle} if, given a graph $G\in\GG$ and integers $k$ and $\sepmin$, in time $T_{\mathcal{S}}(|G|, k, \sepmin)$ it finds a separation $(A,B)$ in $G$
of order at most $\sepord$ with $\sepmin < |A| \leq \sepmax(k, \sepmin)$,
or correctly concludes that $G$ contains a $k$-path.
\end{definition}

For all considered graph classes, we will be able to provide a separation oracle with $\sepmax$ being a polynomial. This, in turn, allows the following generic Turing kernel.

\begin{lemma}\label{lem:generic}
Let $\mathcal{S}$ be a $(T_\mathcal{S}, \sepord, \sepmax)$-separation oracle for a hereditary graph class $\GG$.
Take $\hat{\sepord} := (2h)^{4\sepord+3}$.
Then, the \kPath{} problem restricted to graphs from $\GG$ can be solved:
\begin{itemize}[-,itemsep=4pt]
\item in time $\displaystyle\Oh\big( T_{\mathcal{S}}(|G|, k, k^2 \hat{\sepord}) \cdot |V(G)|\ +\ k \hat{\sepord} \cdot |V(G)|\cdot |E(G)|\big)$,
\item using at most $k \hat{\sepord} \cdot |V(G)|$ calls to \kLinkage{}
\item each call on an induced subgraph of the input graph of size at most
$\sepmax(k, k^2 \hat{\sepord})$.
\end{itemize}
\end{lemma}
\begin{proof}
Let
$$\sepmin = k \cdot \Ibnd(k, \sepord, \sepord) \leq  k(k+1) (\sepord(\sepord+1))^{2\sepord+1} \leq 2k^2 (2\sepord)^{4\sepord+2} \leq  k^2 \hat{\sepord}.$$

As long as $|V(G)| > \sepmin$, we proceed as follows. Invoke algorithm $\mathcal{S}$ on $G$.
If $\mathcal{S}$ claims that $G$ admits a \kPath{}, we simply output the answer yes.
Otherwise, let $(A',B')$ be the separation output by $\mathcal{S}$. 
Apply the Reduction Rule for $k$, $A := A' \setminus B'$, and $Z = N(A) \subseteq A' \cap B'$.
Note that as $|Z| \leq \sepord$, the Reduction Rule deletes at least one vertex of~$A$.
Furthermore, the Reduction Rule invokes
at most
$$\Ibnd(k, \sepord, \sepord) \leq k \hat{\sepord}$$
calls to the oracle,
each call on an induced subgraph of $G$ of size at most 
$$ |A'| \leq \sepmax(k, \sepmin) = \sepmax(k, k^2 \hat{\sepord}).$$
Once we obtain $|V(G)| \leq \sepmin$, we solve the instance using a single call to
\kLinkage{} with $k'=k$, $r=1$, and $R_1 = \emptyset$.
The bounds follow, as there are at most $|V(G)|$ applications of the Reduction Rule,
and each call to the oracle takes $\Oh(|E(G)|)$ time to prepare the instance
and parse the output.
\end{proof}

Note that for any graph class where separations as in Definition~\ref{def:separationOracle} exist, there exists a trivial separation oracle which finds them, running in time
$n^{\sepord+\Oh(1)}$: one iterates over every candidate for $A \cap B$ and, for fixed
set $A \cap B$, a straightforward knapsack-type dynamic programming algorithm checks if one
can assemble $A \setminus B$ of the desired size from the connected components of $G-(A \cap B)$.

However, this running time bound is unsatisfactory, as it greatly exceeds the number of used oracle calls.
For all considered graph classes we prove a much stronger property than just merely the prerequisites of Lemma~\ref{lem:generic}, in particular
providing a more efficient separation oracle. We provide necessary definitions in the next section.

\subsection{Decomposable graph classes}

The following definition captures the key concept of this section.

\begin{definition}
For a constant $\tdadh$ and a computable nondecreasing function $\tdw : \mathbb{Z}_{\geq 0} \to \mathbb{Z}_{\geq 0}$,
a graph class $\GG$ is called \emph{$(\tdw,\tdadh)$-decomposable} if
for every positive integer $k$ and every $G \in \GG$ that does not admit a $k$-path,
the graph $G$ admits a tree decomposition of width less than $\tdw(k)$ and adhesions of size at
most $\tdadh$.
\end{definition}

A standard argument shows that in a decomposable graph class, given the decomposition with
appropriate parameters, it is easy to provide a separation oracle.
\begin{lemma}\label{lem:decomp2sep}
Assume we are given a graph $G$ and a tree decomposition $(T,\X)$ of $G$
of width less than $\tdw$, adhesion at most $\tdadh$, and adhesion degree at most $\tddeg \geq 2$.
Then, given an integer $\sepmin$ such that $|V(G)| > \sepmin$, one can in time $\tdadh^{\Oh(1)} \cdot (|V(G)| + |E(G)| + |V(T)| + \sum_{t \in V(T)} |\X(t)|)$
find a separation $(A,B)$ of order at most $\tdadh$ such that
$$\sepmin < |A| \leq \tdw + \sepmin \cdot \tddeg.$$
\end{lemma}
\begin{proof}
Root the tree $T$ in an arbitrary node, and for $t \in V(T)$ let $T_t$ be the subtree
of $T$ rooted in $t$. 
Let $t_0$ be the lowest node of $T$ such that $|\X(V(T_t))| > \sepmin$; such a node can
be computed in linear time in the size of $G$ and $(T,\X)$.

Group the children $t'$ of $t$ according to their adhesions $\X(t') \cap \X(t)$.
Due to the bound on the adhesion degree, there are at most $\tddeg$ groups. 
For every adhesion $S$, let $X_{t,S}$ be the set of the children of $t$ with
$S = \X(t') \cap \X(t)$. 
Define 
$$V_S = \bigcup_{t' \in X_{t,S}} \X(T_{t'}).$$
We consider now two cases. First, assume that $|V_S| \leq \sepmin$ for every adhesion $S$.
Then, by the adhesion degree bound, we have
$$|\X(T_t)| \leq |\X(t)| + \tddeg \sepmin \leq \tdw + \tddeg \sepmin.$$
Consequently, we can return the separation $(A,B)$ with $A = \X(T_t)$ and $B=\X(T - V(T_t))$.

In the other case, there exists an adhesion $S$ with $V_S > \sepmin$. 
We greedily take a minimal subset $Y_{t,S} \subseteq X_{t,S}$ such that
$$V_S' := \bigcup_{t' \in Y_{t,S}} \X(T_{t'})$$
is of size greater than $\sepmin$.
By the minimality of $t$, for every $t' \in X_{t,S}$ we have $|\X(T_{t'})| \leq \sepmin$
and, consequently $|V_S'| \leq 2\sepmin$.
Thus, we can return the separation $(A,B)$ for $A = V_S'$ and $B = N_G[V(G) \setminus V_S']$,
 as then $A \cap B \subseteq S$. 
\end{proof}

A critical insight 
is that the decomposition used by Cygan et al.~\cite{CyganLPPS14} to solve the \textsc{Minimum Bisection} problem in fact provides an approximate decomposition in a decomposable graph class.
Let us first recall the main technical result of~\cite{CyganLPPS14}.

\begin{definition}\label{def:unbreakable}
A vertex set $X \subseteq V(G)$ of a graph $G$ is called \emph{$(q,\sepord)$-unbreakable} if every separation $(A,B)$ of order at most $\sepord$ satisfies $|(A \setminus B) \cap X| \leq q$ or $|(B \setminus A) \cap X| \leq q$. 
\end{definition}

\begin{theorem}[\cite{CyganLPPS14}]\label{thm:unbreakable}
There is an algorithm that given a graph $G$ and integer $\sepord$ runs in time $2^{\Oh(\sepord^2)} |V(G)|^2 |E(G)|$ and outputs a connected tree decomposition $(T,\Y)$ of $G$ such that: (i)~for each $t\in V(T)$, the bag $\Y(t)$ is $(2^{\Oh(\sepord)}, \sepord)$-unbreakable in $G$, and 
(ii)~for each $tt' \in E(T)$ the adhesion $\Y(t)\cap\Y(t')$ has at most $2^{\Oh(\sepord)}$ vertices and  is $(2\sepord,\sepord)$-unbreakable in $G$.
\end{theorem}

\begin{lemma}\label{lem:unbreakableLocal}
Let $G$ be a graph and suppose there exists a decomposition $(T,\X)$ of $G$ of 
width less than $\tdw$, adhesion $\tdadh$, and adhesion degree $\tddeg$.
Let $(T',\Y)$ be a tree decomposition of $G$ such that for each $t\in V(T')$, the bag $\Y(t)$ is $(2^{\Oh(\tdadh)}, \tdadh)$-unbreakable in $G$.
Then $|\Y(t)| \leq \tdw + \tddeg \cdot 2^{\Oh(\tdadh)}$.
\end{lemma}
\begin{proof}
Consider a bag $Y=\Y(t)$ of the second decomposition, or any $(2^{\Oh(\tdadh)}, \tdadh)$-unbreakable set $Y\subseteq V(G)$.
For every edge $tt' \in E(T)$ of the $\tdadh$-decomposition, removing it partitions $T$ into subtrees $T_{t't}$ and $T_{tt'}$ containing $t'$ and $t$, respectively. The induced separation $(\X(T_{tt'}),\X(T_{t't}))$ on $G$ has order at most $\tdadh$, so one of the sides, say $\X(T_{t't}) \setminus \X(t)$, contains at most $2^{\Oh(\tdadh)}$ vertices of $Y$, by definition of unbreakability. Let us orient the edge $tt'$ away from $t'$ (the `smaller' side). By orienting each edge of $T$ this way, we find a single node $t_0$ of $T$ such that all incident edges point to it. That is, for every neighbor $t'$ of $t_0$, we have $|Y \cap \X(T_{t't_0}) \setminus \X(t_0)| \leq 2^{\Oh(\tdadh)}$.

While $t_0$ can have many neighbors, we can group those neighbors $t'$ according the adhesion $\X(t')\cap \X(t_0)$ to which they correspond.
By the adhesion degree bound, there are at most $\tddeg$ such groups.
If for any such group $U \subseteq N(t_0) \subseteq V(T)$, the union $\bigcup_{t'\in U} \X(T_{t't_0}) \setminus \X(t_0)$ contained more than $3 \cdot 2^{\Oh(\tdadh)}$ vertices of $Y$, then the group can be partitioned into two parts with more than $2^{\Oh(\tdadh)}$ vertices of $Y$ each. This would give a separation of order $\tdadh$ of $G$ with too many vertices of $Y$ on both sides, contradicting its unbreakability.
Therefore, there are at most $\tddeg$ groups, each containing at most $3 \cdot 2^{\Oh(\tdadh)}$ vertices of $Y$, thus the size of $Y$ is bounded by $|\X(t_0)| + \tddeg \cdot 3 \cdot 2^{\Oh(\tdadh)}$
\end{proof}
Thus, the decomposition computed by the algorithm of Theorem~\ref{thm:unbreakable}
approximates the desired decomposition of a decomposable graph class.
\begin{corollary}\label{cor:apx-decomp}
Let $\GG$ be a $(\tdw, \tdadh)$-decomposable graph class.
Then, for every $G \in \GG$ and every integer $k$, one can
in $2^{\Oh(\tdadh^2)} |V(G)|^2 |E(G)|$ time either correctly conclude that $G$ admits a $k$-path,
or find a tree decomposition of $G$
of width at most $(\tdw(k)+1)^{\Oh(\tdadh)}$
and adhesion at most $2^{\Oh(\tdadh)}$.
\end{corollary}
Let us now combine all the above. That is, given an integer $k$ and a graph $G$ from a hereditary $(\tdw, \tdadh)$-decomposable graph class
$\GG$, we start by computing the tree decomposition of Corollary~\ref{cor:apx-decomp} (or conclude there is a $k$-path).
In general this approximated decomposition has adhesion degree $(\tdw(k)+1)^{2^{\Oh(\tdadh)}}$.
We use this decomposition to find separations of any induced subgraphs of $G$ using the algorithm of Lemma~\ref{lem:decomp2sep}  in time $2^{\Oh(\tdadh)}$ times linear in the size of $G$ and the computed decomposition. This gives a $(\tdadh,\sepmax,T)$-separation oracle with $\sepmax(k,\sepmin)=\sepmin \cdot (\tdw(k)+1)^{2^{\Oh(\tdadh)}}$ and $T(n,k,\sepmin)=2^{\Oh(\tdadh)} \cdot n\cdot (\tdw(k)+1)^{\Oh(\tdadh)}$, for any hereditary $(\tdw, \tdadh)$-decomposable graph class.
By plugging it into Lemma~\ref{lem:generic}, we obtain the following.
\begin{corollary}\label{cor:kernel-decomp}
Let $\GG$ be a hereditary $(\tdw, \tdadh)$-decomposable graph class.
Then, the \kPath{} problem, restricted to graphs from $\GG$, can be solved
in time $2^{2^{\Oh(\tdadh)}} |V(G)|^2 |E(G)|$
using
$2^{2^{\Oh(\tdadh)}} kn$
    calls to \kLinkage{} on induced subgraphs of the input graph
of size
$k^2(1 + \tdw(k))^{2^{\Oh(\tdadh)}}$.
\end{corollary}

In the next section, we prove that $H$-minor-free graphs are $(\Oh_H(k^{22}), \Oh_H(1))$-decomposable by analyzing the Global Structure Theorem of $H$-minor-free graphs
due to Robertson and Seymour. A subsequent section provides an analogous result for graphs excluding a fixed topological minor.
In both cases we also get better bounds on the adhesion degree of the approximate decomposition outputted by Theorem~\ref{thm:unbreakable}, improving the bounds in the final kernel.

We would like to remark that we do not want to claim in this paper the idea that, in the context of $H$-(topological)-minor-free graphs, the decomposition of Theorem~\ref{thm:unbreakable}
should be related to the decomposition of the Global Structure Theorem via an argument as in the proof of Lemma~\ref{lem:unbreakableLocal}.
In particular, this observation appeared previously in a work of the second author with Daniel Lokshtanov, Micha\l{} Pilipczuk, and Saket Saurabh~\cite{lpps-manuscript}.

\section{Excluding a minor}\label{sec:minors}
In this section we tackle proper minor-closed graph classes, that is, we prove Theorem~\ref{thm:topo-minor-main} for graph classes excluding a fixed minor, by proving the following.
\begin{theorem}\label{thm:minor-main-aux}
For every fixed graph $H$, the \kPath{} problem restricted to $H$-minor-free graphs
can be solved in time
$\Oh_H(n^2m)$
using $\Oh_H(kn)$
calls to \kLinkage{} on instances being induced subgraphs of the input graph
of size $\Oh_H(k^{24})$.
\end{theorem}

Our main technical result is the following:

\begin{theorem}\label{thm:minor-decomposable}
For every graph $H$, the class of $H$-minor-free graphs is $(\tdw, \tdadh)$-decomposable
for $\tdw(k) = \Oh_H(k^{22})$ and $\tdadh = \Oh_H(1)$.
\end{theorem}

By plugging the above into Corollary~\ref{cor:kernel-decomp}, we obtain the desired
polynomial Turing kernel, but with worse
bounds than promised by Theorem~\ref{thm:minor-main-aux}.
To obtain better bounds, we need to recall the folklore bound on the adhesion degree 
in sparse graph classes; for completeness, we provide a full proof in Appendix~\ref{app:adh-deg-bound}.
\begin{lemma}\label{lem:adh-deg-bound}
Let $G$ be a graph not containing $H$ as a topological minor,
and let $(T,\X)$ be a connected tree decomposition 
of $G$ of width less than $\ell$ and adhesion $h$. 
Then the adhesion degree of $(T, \X)$ is bounded by $f(h, H) \cdot \ell$ for some integer $f(h,H)$ depending only on $h$ and $H$.
\end{lemma}
This way, we conclude that $H$-minor-free graphs without $k$-paths have tree decompositions of width $\Oh_H(k^{22})$, adhesion $\Oh_H(1)$ and adhesion degree $\Oh_H(k^{22})$.
We can use the algorithm of Theorem~\ref{thm:unbreakable} for $\tdadh$ to find (by Lemma~\ref{lem:unbreakableLocal}) an approximate decomposition of width $\Oh_H(k^{22})$, adhesion $\Oh_H(1)$, and thus, again using Lemma~\ref{lem:adh-deg-bound}, of adhesion degree $\Oh_H(k^{22})$.
Theorem~\ref{thm:minor-main-aux} follows from Lemma~\ref{lem:generic}
if we find separations using the algorithm of Lemma~\ref{lem:decomp2sep} applied to this decomposition.

Thus, it remains to prove Theorem~\ref{thm:minor-decomposable}.
For the proof, we use the graph minors structure theorem, decomposing an $H$-minor-free graph~$G$ into parts `nearly embeddable' in surfaces (precise definitions are given in the next subsection). By carefully analyzing details of the structure, we either find a large triconnected embedded part, which must contain a long path by the following theorem of Chen et al.~\cite{ChenYZ12}, or we tighten the graph structure to give a tree decomposition where all parts are small (polynomial in $k$) and adhesions (`boundaries') between them are of constant size.

\begin{theorem}[\cite{ChenYZ12}] \label{thm:embeddedCircumference}
There is a constant $\varepsilon>0$ such that for every integer $t$, every triconnected graph on $n\geq 3$ vertices embeddable in a surface of (Euler) genus $g$ contains a cycle of length at least $n^{\varepsilon} / 2^{(2g+3)^2}$.
\end{theorem}
We note that Chen et al. phrase the theorem (more generally) for $K_{3,t}$-minor-free graphs, but a folklore edge-counting argument shows that graphs embeddable on a surface of genus $g$ are $K_{3,t}$-minor-free for $t=2g+3$ (see e.g.~\cite{Archdeacon96}).

Two intertwined problems that arise with this approach is that torsos of decompositions are not necessarily triconnected, and long paths in them do not necessarily imply long paths in the original graph, because of virtual edges added in torsos.
Torsos can be made triconnected if their near-embeddings include cycles or paths around each vortex, but these may use virtual edges in essential ways. On the other hand, the decomposition can be modified so that virtual edges can be replaced with paths in the original graph, but this requires changes that remove virtual edges, hence potentially removing paths around vortices and destroying triconnectedness.

Because of that, we need to go a little deeper and use a local, strong  version of the structure theorem from Graph Minors XVII~\cite{RobertsonS99}. For the same reason we cannot use existing algorithms for finding the graph minors decompositions. Instead, we only prove the \emph{existence} of a tree decomposition of bounded adhesion, small width, and with nearly embeddable bags.

\subsection*{Global and local graph minor structure theorems}
We now define near-embeddability and the graph minors $\alpha$-decomposition.

\begin{definition}
For an integer $\alpha$, an \emph{$\alpha$-near embedding of $G$} consists of:
\begin{enumerate}[(i)]
	\item a set $A$ of at most $\alpha$ vertices (called the \emph{apex} set);
	\item a family $G_0 \cup G_1 \cup \dots \cup G_{\alpha'} = G\setminus A$ of edge-disjoint subgraphs of $G\setminus A$, where:
	\begin{itemize}
		\item $G_0$ is called the \emph{embedded part},
		\item $\V=\{G_1,\dots,G_{\alpha'}\}$ for some $\alpha'\leq \alpha$ are called (\emph{large}) \emph{vortices},
		\item the intersection $\Omega_i := V(G_0)\cap V(G_i)$ is called the \emph{society} of vortex $G_i$,
		\item vortices are pairwise vertex-disjoint,
	\end{itemize}
	\item an embedding of $G_0$ in a surface of genus at most $\alpha$ such that for $i\in[\alpha']$, the society $\Omega_i$ is embedded on the boundary of a disk whose interior is empty (i.e., does not intersect the embedding or other disks), called the disk \emph{accomodating} $G_i$;
	\item a linear ordering $w_1,\dots,w_{|\Omega_i|}$ of each vortex society $\Omega_i$, corresponding to its natural ordering around its disk (for some choice of direction and starting point);
	\item for each large vortex $G_i \in \V$, a path decomposition $\Z_1,\dots,\Z_{|\Omega_i|}$ of $G_i$ of width at most $\alpha$ such that $w_j \in \Z_j$, for $j \in [|\Omega_i|]$.
\end{enumerate}
We denote such an $\alpha$-near embedding as $(A,G_0,\V)$, with the embedding and path decompositions only implicit in the notation.
\end{definition}

\begin{definition}
A \emph{(graph minors)} \emph{$\alpha$-decomposition} of a graph $G$ consists of:
\begin{itemize}
\item a rooted tree decomposition $(T,\Y)$ of $G$ of adhesion at most $\alpha$;
\item for each $t\in V(T)$, an $\alpha$-near embedding ($A_t,G_{t0},\V_t)$ of $\torso(G,\X(t))$.
\end{itemize}
%such that furthermore, for every edge $pt$ of $T$ where $p$ is the parent of $t$, the adhesion $\X(p)\cap \X(t)$ is contained in $A_t$, while $(\X(p)\cap \X(t)) \setminus A_p$ is either contained in a single bag of a path decomposition of some vortex in $\V_p$, or has at most 3 vertices, contained in a single face of $G_{p0}$.
\end{definition}

In this subsection we prove the following variant of the Global Structure Theorem which
implies Theorem~\ref{thm:minor-decomposable}.

\begin{theorem}\label{thm:global}
	For every graph $H$, there is a constant $\alpha = \alpha(H)$ such that the following holds, for any integer $k$:
	any graph $G$ excluding $H$ as a minor and without a $k$-path has an $\alpha$-decomposition of width at most $\alpha \cdot k^{22}$. 
\end{theorem}

We deduce Theorem~\ref{thm:global} from a similarly modified variant of the Local Structure Theorem.
An \emph{$\alpha$-near local embedding} of a graph $G$ is defined similarly to an $\alpha$-near embedding, but we allow an arbitrary number of `small vortices' $\W$ and we allow the path decompositions of all vortices to have only bounded adhesion instead of bounded width; thus arbitrarily complicated graphs can hide behind vortices. 
We additionally require each large vortex to be surrounded by a certain path.
Formally:

\begin{definition}
A \emph{comb} is a union of a path $P$ with some mutually vertex-disjoint (possibly trivial) paths that have their first vertex and no other vertex on $P$. The last vertices of those paths are called the \emph{teeth} of the comb, they are naturally ordered by $P$.
\end{definition}

\newcounter{localenumi}

\begin{definition}
For an integer $\alpha$, an \emph{$\alpha$-near local embedding of $G$} consists of:
\begin{enumerate}[(i),series=localStructure]
	\item a set $A$ of at most $\alpha$ vertices (called the \emph{apex} set);
	\item a family $G_0 \cup G_1 \cup \dots \cup G_n = G\setminus A$ of edge-disjoint subgraphs of $G\setminus A$, where:
	\begin{itemize}
		\item $G_0$ is called the \emph{embedded part},
		\item $\V=\{G_1,\dots,G_{\alpha'}\}$ for some $\alpha'\leq \alpha$ are called \emph{large vortices},
		\item $\W=\{G_{\alpha'+1},\dots,G_n\}$ are called \emph{small vortices},
		\item vortices intersect only in $G_0$: $V(G_i) \cap V(G_j) \subseteq V(G_0)$ for $i\neq j \in [n]$,
		\item the intersection $\Omega_i := V(G_0)\cap V(G_i)$ is called the \emph{society} of vortex $G_i$,
		\item large vortices are pairwise vertex-disjoint,
		\item small vortices have societies of size $\leq 3$;
	\end{itemize}
	\item an embedding of $G_0$ in a surface of genus at most $\alpha$ such that for $i\in[n]$, the society $\Omega_i$ is embedded on the boundary of a disk whose interior is empty (i.e., does not intersect the embedding or other disks), called the disk \emph{accomodating} $G_i$;
	\item a linear ordering $w_1,\dots,w_{|\Omega_i|}$ of each society $\Omega_i$, corresponding to its natural ordering around its disk (for some choice of direction and starting point);
	\item\label{p:pathwidth} for each large vortex $G_i \in \V$, a path decomposition $\Z_1,\dots,\Z_{|\Omega_i|}$ of $G_i$ of adhesion at most $\alpha$ such that $\Z_j \cap \Omega_i = \{w_{j-1},w_j\}$ for $1 < j \leq |\Omega_i|$ and $\Z_1 \cap \Omega_i = \{w_1\}$.
	%\item linkage
	\item\label{p:comb} for each large vortex $G_i \in \V$, a comb $C_i$ in $G_i \cup G_0 \cup \bigcup \W$ whose teeth are vertices  $w_1,\dots,w_{|\Omega_i|}$ of $\Omega_i$, in the same order.
\end{enumerate}
We denote such an $\alpha$-near local embedding as $(A,G_0,\V,\W)$, with the embedding and path decompositions only implicit in the notation.
The near embedding is said to \emph{respect} a tangle $\T$ if the large side of a separation in $\T\setminus A$ is never contained in a vortex $W \in \W$ or in a bag of the decomposition of a vortex $V \in \V$.
\end{definition}

\begin{definition}
In a graph $G$, a tangle $\T$ \emph{controls an $H$-minor} if there is a minor model in $G$ (defined by branch sets $(B_h)_{h \in V(H)}$ inducing vertex-disjoint connected subgraphs in $G$ and with an edge between $B_h$ and $B_{h'}$ whenever $hh' \in E(H)$) such that no branch set is fully contained in a small side of a separation in $\T$.
\end{definition}

The following Local Structure Theorem follows from~\cite{RobertsonS99}, as explained in~\cite{FrohlichM11} (we note that while~\cite{FrohlichM11} assumes that $G$ is $H$-minor-free, the original statement in~\cite{RobertsonS99} only requires that $\T$ controls no $H$-minor -- we will need this stronger version when dealing with topological minors).
Since this is crucial for our approach, we stress that the comb $C_i$ contains a path of length at least $|\Omega_i|$ (a scrupulous reading of \cite{RobertsonS99} implies even stronger statements, but this will suffice).

\begin{theorem}[\cite{FrohlichM11}]\label{thm:GMlocal}
	For every graph $H$ there exist integers $\alpha,\theta$ such that: for every graph $G$ and every tangle $\T$ in $G$ of order $\geq \theta$ that controls no $H$-minor, there is an $\alpha$-near local embedding of $G$ which respects $\T$.
\end{theorem}

We now improve the statement of the  Local Structure Theorem~\ref{thm:GMlocal} in steps: first requiring small vortices to be `well attached', then making the torsos triconnected, and finally deducing bounds on width in terms of longest path length.

\begin{lemma}\label{lem:properlyAttached}
	In Theorem~\ref{thm:GMlocal}, we can additionally assume that:
	\begin{enumerate}[(i),resume=localStructure]
		\item\label{p:wellAttached} for each small vortex $G_i \in \W$ and every $v_1,v_2 \in \Omega_i$, there is a path in $G_i$ between $v_1$ and $v_2$, with no internal vertices in $\Omega_i$.
	\end{enumerate}
\end{lemma}
\begin{proof}
	Property~\ref{p:wellAttached} is guaranteed by (9.1) in~\cite{RobertsonS99}. (As explained in~\cite{FrohlichM11}, the border cells of~\cite{RobertsonS99} are translated into bags of large vortex decompositions, so small vortices arise only from internal cells; for an internal cell $c_0$, $\alpha(c_0)^*$ is translated into a small vortex $G_i$, while $\gamma(\tilde c_0)$ is translated into $\Omega_i$, giving exactly the statement we want. We also note that the proof of (9.1) simply partitions a vortex with no $v_1$-$v_2$ path into two vortices with at most two vertices in their societies).
\end{proof}

\begin{definition}\label{def:G0star}
	For an $\alpha$-near local embedding $(A,G_0,\V,\W)$ of $G$, define $G_0^*$ to be the graph obtained from $G_0$ by adding an edge between every two consecutive vertices in each society $\Omega_i$, and a new vertex for each society of a large vortex, with edges to every vertex of the society.
	The embedding of $G_0$ is naturally extended to an embedding of $G_0^*$ (the new vertices and edges embedded in place of the accommodating disks).
\end{definition}

\begin{lemma}\label{lem:localConnected}
	In Lemma~\ref{lem:properlyAttached}, we can additionally assume that:
	\begin{enumerate}[(i),resume=localStructure]
		\item\label{p:triconnected} $G_0^*$ is triconnected.
	\end{enumerate} 
\end{lemma}
\begin{proof}
	Suppose $G_0^*$ has a separation $(X,Y)$ of order $\leq 2$.
	Since societies of small vortices induce cliques in $G_0^*$, they are contained in $X$ or $Y$. 
	Similarly for each large vortex $G_i \in \V$, its society $\Omega_i$, together with the new vertex in $G_0^*$ with $\Omega_i$ as neighborhood, induces a wheel in $G_0^*$ , which is triconnected, hence contained in $X$ or $Y$.
	In the language of~\cite{RobertsonS99}, this translates back into a partition $I$, $J$ of cells, whose intersection corresponds to at most two embedded vertices $v,v'$ (that is, $\gamma(I) \cap \gamma(J) = \{v, v'\}$).
	But this contradicts (11.1) in~\cite{RobertsonS99}, which states that such an intersection has size at least 3.
\end{proof}

\begin{lemma}\label{lem:localSmall}
	In Lemma~\ref{lem:localConnected}, we can additionally assume that there is a constant $\alpha'$ depending on $H$ only such that:
	\begin{enumerate}[(i),resume=localStructure]
		\item\label{p:size} for any integer $k$, if $G$ does not contain a path of length $\geq k$, then $|V(G_0)| < \alpha'\cdot k^{22}$.
	\end{enumerate} 
\end{lemma}
\begin{proof}
For a given graph $H$, let $\alpha$ be the constant given by Theorem~\ref{thm:GMlocal}.
By Theorem~\ref{thm:embeddedCircumference}, there is a universal constant $\varepsilon$ such that any triconnected graph $G_0$ embeddable in a surface of genus at most $\alpha$ contains a cycle of length at least $|G_0|^\varepsilon / 2^{(2\alpha+2)^2}$.
Let $c := \frac{2}{\varepsilon}$ and $\alpha' := (\alpha \cdot 2^{(2\alpha+2)^2})^{1/\varepsilon}$.
For an $H$-minor-free graph $G$, suppose the $\alpha$-near local embedding given by Lemma~\ref{lem:localConnected} has $|V(G_0)| \geq \alpha' \cdot k^{c}$. Then $G_0^*$ is a triconnected graph (by property~\ref{p:triconnected}) embedded in a surface of genus at most $\alpha$ (by Definition~\ref{def:G0star}) with at least as many vertices, so $G_0^*$ must contain a cycle $C$ of length at least $(\alpha' k^{c})^\varepsilon / 2^{(2\alpha+2)^2} = \alpha \cdot k^2$.

If the society $\Omega_i$ of any vortex of the $\alpha$-near local embedding of $G$ has at least $k$ vertices, then $G$ contains a path of at least that length (in the comb from property~\ref{p:comb}), in which case the lemma follows.
Otherwise, there are at most $\alpha$ large vortices, hence at most $\alpha \cdot (k-1)$ vertices in their societies, and at most~$\alpha$ virtual vertices that were inserted as the centers of the wheels. 
Let $P$ be the longest subpath of $C$ between any two such vortex-related vertices (or any subpath not visiting any such vertices, if there are less than two). Then $P$ is a path of length at least $\frac{\alpha \cdot k^2}{\alpha \cdot k} = k$ in $G_0^*$ whose edges in $E(G_0^*)\setminus E(G_0)$ came only from small vortices, by the definition of $G_0^*$ (Definition~\ref{def:G0star}). By property~\ref{p:wellAttached}, these edges can be replaced with paths (each of length at least 2) in the corresponding small vortices, giving a path in $G$. 
Only one or two consecutive edges can come from the same small vortex (since their societies have at most 3 vertices), so the resulting path has length at least $k$.

The constant from Theorem~\ref{thm:embeddedCircumference} given by Chen et al.~\cite{ChenYZ12} is $\varepsilon = \log_{1729}2$, hence the universal constant here is $c=\frac{2}{\varepsilon} = 2 \log_2(1729) < 22$.
\end{proof}

Finally, we deduce Theorem~\ref{thm:global} (the Global Structure Theorem) from Theorem~\ref{thm:GMlocal} (the Local Structure Theorem) by a standard induction, exactly as done by Diestel et al.~\cite[Theorem 4]{DiestelKMW12}. (Note that in case the excluded graph $H$ is planar, we can already conclude the theorem trivially from the fact that there is a tree-decomposition of width $\Oh_H(1)$~\cite{RobertsonS86}).
The only difference is that in the Local Structure Theorem we add the bound on the size of the embedded part from Lemma~\ref{lem:localSmall}.
In the proof from~\cite{DiestelKMW12}, every bag of the created decomposition is either constructed as a set of size bounded as $\Oh_H(1)$, or is constructed from an $\alpha$-near local embedding by taking the vertices of: the embedded part $G_0$, the set of apices $A$, and for each large vortex, the intersection of every two consecutive bags of its path decomposition, each of size at most $\alpha=\Oh_H(1)$.
The number of bags in a path decomposition of a large vortex is equal to the size of its society, and large vortex societies are disjoint subsets of $V(G_0)$. Therefore every tree decomposition bag constructed in the proof has size at most $|V(G_0)| + \alpha + \alpha \cdot |V(G_0)| = \Oh_H(k^{22})$, by Lemma~\ref{lem:localSmall}.
This proves the additional condition we require in Theorem~\ref{thm:global}.

%\begin{theorem}
%	There is a constant $c<22$ such that for every graph $H$ and every integer $m$, there is a constant $\alpha(H)$ such that the following holds, for any integer $k$:
%	for any graph $G$ excluding $H$ as a minor and without a $k$-path and any subset  $Z\subseteq V(G)$ of at most $m$ vertices, $G$ has an $\alpha$-decomposition of width  at most $\alpha \cdot k^c$, with $Z$ contained in the apex set of the root (i.e, of the near embedding of the root's torso given by the $\alpha$-decomposition).
%\end{theorem}
%\begin{proof}
%	Consider any $H$, $m$. Theorem~\ref{thm:GMlocal} and Lemma~\ref{lem:localSmall} provide integers $\alpha$, $\theta$ and $\alpha'$. We may assume $m$ to be large enough so that $\theta'' := (m+2)/3 \geq \max(3\hat{\alpha}+3, \theta)$ and $\theta''$ is integral. Let $\alpha'' := 4 \theta -2$.
%
%	For these now fixed $R,m,\alpha''$, we prove the theorem by induction on $|G|$.
%	Assume $|Z|=m$, otherwise we add arbitrary vertices to it (or conclude the theorem trivially if there are no more vertices, implying $|G|<m<\alpha''$).
%	
%\end{proof}

\section{Excluding a topological minor}\label{sec:topological}
In this section we tackle graph classes excluding a topological minor, that is, we prove Theorem~\ref{thm:topo-minor-main} by proving the following.

\begin{theorem}\label{thm:topo-main-aux}
For every fixed graph $H$, the \kPath{} problem restricted to $H$-topological-minor-free graphs
can be solved in time
$\Oh_H(n^2m)$
using $\Oh_H(kn)$
calls to \kLinkage{} on instances being induced subgraphs of the input graph
of size $k^{\Oh_H(1)}$.
\end{theorem}

This follows as before from the following decomposability theorem. 
Note the exponent in the polynomial bound on width (bag size) now depends on $H$.

\begin{theorem}\label{thm:topo-decomposable}
For every graph $H$, the class of $H$-topological-minor-free graphs is $(\tdw, \tdadh)$-decomposable
for $\tdw(k) = k^{\Oh_H(1)}$ and $\tdadh = \Oh_H(1)$.
\end{theorem}

We proceed with the proof of the above theorem.
Grohe and Marx~\cite{GroheM15} proved that when excluding a topological minor, graphs admit a similar structure as for excluding a minor, but apart from nearly embeddable parts, one needs to consider parts that have bounded degree except for a bounded number of vertices:

\begin{definition}
A graph $G$ has \emph{$(a,d)$-bounded degree} (for $a,d \in \mathbb{N}$)
if $\leq a$ vertices of $G$ have degree $>d$ in $G$.
%An \emph{$\alpha$-topological decomposition} of a graph $G$ is a rooted tree decomposition $(T,\X)$ of $G$ of adhesion at most $\alpha$ such that for each $t\in V(T)$, either $\torso(G,\X(t))$ has $(\alpha,\alpha)$-bounded degree, or it has an $\alpha$-near embedding.
\end{definition}

The main `Global Structure Theorem' of~\cite{GroheM15} is that graphs excluding a topological minor $H$ admit a tree decomposition of bounded adhesion, where the torso of every bag is either nearly embeddable (as when excluding a minor), or has $(h,d)$-bounded degree, for $d=\Oh_H(1)$.
We still need the slightly stronger notion of near embeddability (with combs) for our proofs, which is why we will work with the following definition of `Local Structure', analogous to $\alpha$-near local embeddings.

A \emph{star decomposition} of a graph is a (rooted) tree decomposition $(\T,\X)$ where $T$ is a star: a root node $t_{\text{root}}$ with only leaf nodes attached.
We call the bag $\X(t_{\text{root}})$ the \emph{root bag} and the remaining ones \emph{leaf bags}.
The \emph{root torso} is $\torso(G, \X(t_{\text{root}}))$. 
A star decomposition \emph{respects} a tangle $\T$ if the large side of any separation in $\T$ is not fully contained in any single leaf bag.

Grohe and Marx~\cite[Lemma 6.12]{GroheM15} proved the following Local Structure Theorem.%
\footnote{Grohe and Marx~\cite{GroheM15} use a slightly stronger definition of \emph{respecting a tangle}, which is however equivalent to the conjunction of this one and the assertion that all adhesions in the star decomposition are strictly smaller than the order of $\T$, which we prefer to state separately. They also assume a slightly weaker condition than that of a tangle controlling a minor (which they call \emph{not being removed by} the tangle), but it is trivially implied.}

\begin{theorem}[\cite{GroheM15}]\label{thm:groheMarx}
For every integer $k$ there exist integers $\theta,d,k'$ such that:
for every graph $G$ excluding $K_k$ as a topological minor and every tangle $\T$ of order $\geq\theta$ that controls a $K_{k'}$-minor,
there is a star decomposition of $G$ that respects $\T$, has adhesion $<\theta$, and such that the root torso has $(k,d)$-bounded degree.
\end{theorem}

\begin{corollary}\label{cor:groheMarx}
For every graph $H$ there exist integers $\theta,\alpha,d$ such that:
for every graph $G$ excluding $H$ as a topological minor and every tangle $\T$ in $G$ of order $\geq\theta$, $G$ has one of the following:
\begin{itemize}
\item a star decomposition that respects $\T$, has adhesion $\alpha$, and such that the root torso has $(|V(H)|,d)$-bounded degree, or
\item an $\alpha$-near local embedding that respects $\T$, satisfying properties~\ref{p:wellAttached}-\ref{p:size}.
\end{itemize}
\end{corollary}
\begin{proof}
	For $k=|V(H)|$ let the constants given by Theorem~\ref{thm:groheMarx} be $\theta',d',k'$.
	Let $\alpha'',\theta''$ be the constants given by Theorem~\ref{thm:GMlocal} for graphs excluding $K_{k'}$ as a minor.
	We prove the claim for $\theta = \max(\theta',\theta'')$, $\alpha = \max(\theta',\alpha'')$ and $d=d'$.
	If the tangle $\T$ of $G$ controls no $K_{k'}$-minor, then the claim follows from Theorem~\ref{thm:GMlocal}. Otherwise $\T$ does control an $K_{k'}$-minor and Theorem~\ref{thm:groheMarx} applies.
\end{proof}

The main theoretical tool we use, besides the Grohe and Marx decomposition in the form of Corollary~\ref{cor:groheMarx}, is the following lower bound on cycles in graphs of bounded degree by Shan~\cite{Shan15} (we note the constant $\Delta-1$ is asymptotically optimal; for small $\Delta$, previous results by Chen et al.~\cite{ChenGYZ06} give a slightly better bound). %of $n^{1/\log_2(\max(4\Delta+1,64))}/2 + 3$).
\begin{theorem}[\cite{Shan15}]\label{thm:degreeCircumference}
If G is a triconnected graph with maximum degree at most $\Delta \geq 425$, then $G$ has a cycle of length at least $n^{1/\log_2(\Delta-1)}/4 + 2$.
\end{theorem}

With this in hand, we can proceed to bounding the bags of almost bounded degree.

\begin{lemma}\label{lem:topoLocalStructure}
	In Corollary~\ref{cor:groheMarx}, we can additionally assume that in the first outcome, for any integer $k$, if $G$ has no $k$-path, then the root bag has size $k^{\Oh_H(1)}$.
\end{lemma}
\begin{proof}

Without loss of generality assume that $\theta \geq \alpha (|V(H)|+2)$, in the statement of Corollary~\ref{cor:groheMarx} (otherwise increase $\theta$).
Let us first ensure that the star decomposition $(T,\X)$ is connected.

\begin{numberedclaim}
	We can additionally assume that
	\begin{align}
			&\text{For every leaf node $t \in V(T)$, the set $\X(t) \setminus \X(t_{\text{root}})$ is connected, and}\label{eq:C1}\tag{C1}\\
			&\text{$\X(t) \setminus \X(t_{\text{root}})$ has edges to each vertex in the adhesion $\X(t) \cap \X(t_{\text{root}})$}\label{eq:C2}\tag{C2}
	\end{align}
\end{numberedclaim}
\begin{claimproof}
	If \eqref{eq:C1} fails we can partition the leaf bag into many leaf bags.
	If \eqref{eq:C2} fails we can remove the adhesion vertex in question from the leaf bag.
\end{claimproof}

\newcommand{\Groot}{\ensuremath{G_{\text{root}}}}

We need to define a variant of the torso that is similarly connected as the root torso, but has no virtual edges (so that long paths in it imply long paths in $G$).
Construct $\Groot^*$ from $G$ as follows: contract $\X(t)\setminus \X(t_{\text{root}})$ to a single vertex $x_t$ for each leaf $t\in T$ and merge any such vertices $x_t, x_{t'}$ that have equal neighborhoods (corresponding to the same adhesion) into one.
Let $X$ be the set of all those new vertices.

\begin{numberedclaim}
	If $G$ has no $k$-path, then neither does $\Groot^*$ .
\end{numberedclaim}
\begin{claimproof}
For each leaf node~$t \in V(T)$, the set~$\X(t) \setminus \X(t_{\text{root}})$ is connected by~\eqref{eq:C1} and can therefore be contracted to a single vertex~$x_t$ by repeated edge contractions. Since~$\Groot^*$ is a subgraph of the graph obtained by these edge contractions, it follows that~$\Groot^*$ is a minor of~$G$. As minor operations do not increase the length of a longest path, the claim follows.
	%Let $P$ be a path of length at least $k$ in $\Groot^*$.
	%Recall that $\Groot^*$ can be obtained from the subgraph of $G$ induced by the root bag of $(T,\X)$ by adding the vertices of $X$.
	%Since each vertex $x_t \in X$ comes from contracting a set $\X(t)\setminus \X(t_{\text{root}})$ (for some $t \in V(T)$) to a single vertex, and since by~\eqref{eq:C1} this is a connected set in $G$ with neighborhood $\X(t)\cap \X(t_{\text{root}})$ equal to the neighborhood of $x_t$, by~\eqref{eq:C2}, any vertex of $P$ in $X$ can be replaced (together with incident edges) by a path in $G[\X(t)\setminus \X(t_{\text{root}})]$.
	%Different elements of $X$ come from sets  $\X(t)\setminus \X(t_{\text{root}})$ with different $t\in V(T)$, which are pairwise disjoint, so all the replacements can be applied to give a path in $G$.
	%Each time we replace at most two incident edges with a path that contains at least one internal vertex in $\X(t)\setminus \X(t_{\text{root}})$, hence the length of the path in $G$ is at least $k$.
\end{claimproof}

\begin{numberedclaim}
	 $\Groot^*$ has $(|V(H)|,\max(\alpha, d + 2^d))$-bounded degree.
\end{numberedclaim}
\begin{claimproof}
	Observe that $\Groot^*$ can be obtained from the root torso of $(T,\X)$ by removing virtual edges and adding the vertices of $X$. Each vertex in $X$ has degree equal to the size of some corresponding adhesion, which is $\leq \alpha$.
	Each vertex $v$ in the torso gains as neighbors only vertices in $X$, which correspond to different (as vertex subsets) adhesions containing $v$; since these adhesions induced cliques containing $v$ in the torso (which had $(|V(H)|,d)$-bounded degree), their number is bounded by $2^d$ (unless $v$ was already one of the $\leq |V(H)|$ vertices of high degree).
\end{claimproof}

\begin{numberedclaim}
	If $\Groot^*$ does not contain a path of length at least $k$, then $\Groot^*$ has a tree decomposition of adhesion $\leq |V(H)| + 2$ and width $k^{\Oh_H(1)}$.
\end{numberedclaim}
\begin{claimproof}
	Assume that~$\Groot^*$ does not contain a path of length at least~$k$. By the previous claim, in $\Groot^*$ one can delete $\leq |V(H)|$ vertices to obtain a subgraph of degree bounded by $\Delta := \max(\alpha, d+2^d)$; name this subgraph $\Groot^\dagger$.
	Consider the Tutte decomposition~\cite{Tutte66} of $\Groot^\dagger$. In modern terms (cf.~\cite[Theorem 1]{Jansen17}), it is a tree decomposition of~$\Groot^\dagger$ of adhesion at most two, whose torsos are triconnected topological minors of~$\Groot^\dagger$. Consequently, each torso again has degree bounded by $\Delta$ and cannot contain a path of length at least $k$.
	By Theorem~\ref{thm:degreeCircumference}, if $n$ denotes the size of such a bag, then $n^{\varepsilon}/4+2 < k$, where $\varepsilon = 1/\log_2(\max(425,\alpha,d+2^d)-1)$.
	Thus the size of each bag is $n < (4(k-2))^{1 / \varepsilon} = k^{\Oh(d + \log_2 \alpha)} = k^{\Oh_H(1)}$.
	Adding the $\leq |V(H)|$ deleted vertices back to every bag, we obtain a tree decomposition of $\Groot^*$ as claimed.
\end{claimproof}

Denote the root torso of $(T,\X)$ as $\Groot$.

\begin{numberedclaim}
	For every tree decomposition of $\Groot^*$, one can find a tree decomposition of $\Groot$ with adhesion and maximum bag size at most $\alpha$ times larger than the adhesion and maximum bag size of the original tree
  decomposition.
\end{numberedclaim}
\begin{claimproof}
	Let $(S,\Y)$ be a tree decomposition of $\Groot^*$.	
	Define $\Y'(s)$ as $\Y(s)$ with vertices $x_t \in X$ replaced by all of $N_{\Groot^*}(x_t)$.
	We claim $(S,\Y')$ is a tree decomposition of $\Groot$.
	Every edge of $\Groot$ is either an edge of $\Groot^* - X$, still contained in the same bag, or a virtual edge between two vertices of an adhesion, hence contained in  $N_{\Groot^*}(x_t)$ for some $x_t \in X$ and thus in a bag $\Y'(s)$ (for an $s$ such that $\Y(s)$ contained $x_t$).
	To show that the subset $\Y'^{-1}(v) := \{s\in S : v \in \Y'(s)\}$ is connected, for any given vertex $v \in V(\Groot)$, observe that it is the union of $\Y^{-1}(v)$ with $\Y^{-1}(x_t)$ for all $x_t \in X$ adjacent to $v$ in $\Groot^*$;
	since such sets $\Y^{-1}(v)$ and $\Y^{-1}(x_t)$ are connected in $S$ and intersect (as $x_t$ is adjacent to $v$), their union is connected.
	Each vertex in each bag has been replaced by at most $|N_{\Groot^*}(x_t)|$ vertices, which is the size of an adhesion in $(T,\X)$ and hence at most~$\alpha$. It follows that the maximum bag size increases by a factor at most~$\alpha$.

	To bound any adhesion, observe similarly that  for $s_1,s_2\in V(S)$, the adhesion equals 
		$$\Y'(s_1) \cap \Y'(s_2) = \left((\Y(s_1) \cap \Y(s_2)) \setminus X\right) \cup \bigcup_{x_{t_1} \in \Y(s_1), x_{t_2} \in \Y(s_2)} N(x_{t_1}) \cap N(x_{t_2}).$$
	For any vertex $v \in N(x_{t_1}) \cap N(x_{t_2})$ (for some $x_{t_i} \in \Y(s_i)$), the set $\Y^{-1}(v)$ is connected and must intersect both $\Y^{-1}(x_{t_1})$ and $\Y^{-1}(x_{t_2})$. Hence either 
	$\Y^{-1}(x_{t_1})$ contains both $s_1$ and $s_2$, or $\Y^{-1}(x_{t_2})$ does, or $\Y^{-1}(v)$ does.
	This means $x_{t_1}$, $x_{t_2}$ or $v$ is contained in $\Y(s_1) \cap \Y(s_2)$.
	Therefore $$\Y'(s_1) \cap \Y'(s_2) \subseteq \left((\Y(s_1) \cap \Y(s_2)) \setminus X\right) \cup \bigcup_{x_t \in \Y(s_1) \cap \Y(s_2)} N(x_t).$$
	This again implies each adhesion increases at most $\alpha$ times.
\end{claimproof}

By the last three claims, assuming $G$ has no $k$-path, $\Groot$ has a tree decomposition $(S,\Y)$ of adhesion $\alpha' := \alpha (|V(H)|+2)$ and width $k^{\Oh_H(1)}$.
We focus this decomposition on the unique bag that, in some sense, respects the tangle $\T$.
Let us first see how $\T$ orients separations of $\Groot$. 
For a separation $(A,B)$ of $\Groot$, every adhesion $\X(t)\cap \X(t_{\text{root}})$ of $(T,\X)$ induces a clique in $\Groot$, hence it is fully contained in $A$ or in $B$; thus $(A,B)$ can be extended to a separation of $G$ by adding all of $\X(t)$ to $A$ or $B$, correspondingly, with an arbitrary choice for $\X(t)$ such that $\X(t)\cap \X(t_{\text{root}}) \subseteq A \cap B$.
We call such separations of $G$ \emph{extensions} of the separation $(A,B)$ of $\Groot$.
For a separation $(A,B)$ of order $\leq\theta$ of $\Groot$, observe that either all extensions of $(A,B)$ and no extensions of $(B,A)$ belong to $\T$, or vice-versa: this follows from tangle axioms and the fact that $(T,\X)$ respects $\T$ (that is, each $\X(t)$ such that $\X(t)\cap \X(t_{\text{root}}) \subseteq A \cap B$ is small w.r.t. $\T$, so putting it one side of an extension $(A',B')$ instead of the other cannot change whether $A'$ is small, as otherwise the two small sides together with $\X(t)$ would cover all of $G$).
%Indeed, an extension $(A',B')$ of $(A,B)$ has order $\leq \theta$, so either $(A',B') \in \T$ or $(B',A')\in\T$. If $A'$ is the small side ($(A',B')\in \T$), say, then consider any $\X(t)$ such that $\X(t)\cap \X(t_{\text{root}}) \subseteq A \cap B$. If the extension has been chosen so that $\X(t) \subseteq A'$, then moving it to the other side we obtain a different extension $(A'',B'') := (A' \setminus (\X(t)\setminus \X(t_{\text{root}})), B' \cup \X(t))$ which must still belong to $\T$, since $G[B''] \cup G[A'] \supseteq G[B'] \cup G[A'] = G$. If on the other hand $\X(t) \subseteq B'$, then again by moving it to the other side we obtain an extension $(A'',B'') := (A' \cup \X(t), B'\setminus (\X(t)\setminus \X(t_{\text{root}})))$ which must still belong to $\T$, since otherwise $G[B''] \cup G[\X(t)] \cup G[A'] = G$ would contradict that $\T$ is a tangle.

We now orient edges of the decomposition $(S,\Y)$ according to $\T$, to find a (unique) bag we can focus on.
For each edge $s_1 s_2 \in E(S)$, if $S_1, S_2$ denote the two connected components of $S - s_1s_2$ containing $s_1, s_2$, respectively, then $(\Y(S_1),\Y(S_2))$ is a separation with $|\Y(S_1) \cap \Y(S_2)| = |\Y(s_1) \cap \Y(s_2)| \leq \alpha' \leq \theta$, hence either all of its extensions or all extensions of $(\Y(S_2),\Y(S_1))$ belongs to the tangle $\T$.
We direct the edge $s_1 s_2$ to either $s_2$ or $s_1$ accordingly (towards the side with large extensions). After repeating this for every edge of $S$, we find a node $s_0$ with no outgoing edges.
This means that if we root $S$ at $s_0$ and $S_1,S_2,\dots,S_\ell$ are the subtrees of $S-s_0$, then all extensions of $(\Y(S_i), \Y(S - S_i))$ belong to $\T$, for $i \in [\ell]$.

We can now forget about some details of the tree decomposition and define the following star decomposition $(S',\Y')$ of $\Groot$: its nodes are $s_{\text{root}} := s_0$ and its direct children $s_1,\dots,s_\ell$ in $S$, while its bags are $\Y'(s_{\text{root}}) := \Y(s_0)$ and $\Y'(s_i) := \Y(S_i)$ for $i \in [\ell]$.
By construction $(S',\Y')$ is a star decomposition of $\Groot$ of adhesion $\alpha'$, with a root bag of size $k^{\Oh_H(1)}$, which satisfies the following:
\begin{equation}\label{eq:respect}
 \text{all extensions of }(\Y(s_i), \Y(S-s_i))\text{ are in }\T\text{, for }i \in [\ell].
\end{equation}

The following claim thus concludes the proof of the lemma, by extending this star decomposition of $\Groot$ to a star decomposition of $G$ respecting $\T$, with a root torso of almost bounded degree and size $k^{\Oh_H(1)}$.

\begin{numberedclaim}
	If $\Groot$ has a star decomposition  $(S,\Y)$ of adhesion $\alpha' \leq \theta$ satisfying~\eqref{eq:respect}, then $G$ has a star decomposition of adhesion $\max(\alpha,\alpha')$ respecting $\T$ where the root torso has $(|V(H)|,d+\alpha'\cdot d)$-bounded degree and its vertex set is equal to the root bag of $(S,\Y)$.
\end{numberedclaim}
\begin{claimproof}
	%If $\alpha< 2$ just increase it to 2. We now modify the decomposition to make it $\alpha+1$-connected, without introducing adhesions larger than $\alpha$.
	Let $(S,\Y)$ be a star decomposition of $\Groot$ satisfying~\eqref{eq:respect}. 
	Without loss of generality assume that for every leaf $s\in V(S)$, every vertex of the adhesion $\Y(s) \cap \Y(s_{\text{root}})$ has a neighbor in $\Y(s) \setminus \Y(s_{\text{root}})$ (otherwise we can remove it from $\Y(s)$).
	Every adhesion of the (larger) star decomposition $(T,\X)$ of $G$ induces a clique in $\Groot$ by definition, hence we can choose for each $t \in V(T)$ a node $s(t) \in V(S)$ such that $\X(t) \cap \X(t_{\text{root}}) \subseteq \Y(s(t))$.
	We construct a new star decomposition $(T',\X')$ of $G$ from $(T,\X)$ as follows:
	the root bag becomes $\Y(s_{\text{root}})$, leaves $t$ with $s(t) = s_{\text{root}}$ and their bags remain unchanged, the remaining leaves are deleted, and each leaf $s$ of $(S,\Y)$ is added as a new leaf of $T'$, with $\X'(s) := \Y(s) \cup \bigcup_{t \colon s(t) = s} \X(t)$.

	It is straightforward to check that $(T',X')$ is a star decomposition of $G$ of adhesion $\max(\alpha,\alpha')$, as follows.
	Every edge is covered by a bag in $\X'$, because it was covered by a bag in $\X$, which is either a subset of a bag in $\X'$, or in $\X(t_{root})$, and hence in a bag of $\Y$, which are again subsets of bags in $\X'$.
	Every vertex occurs in bags $\X'$ of a connected subtree of $T'$, because if a vertex occurs in two different leaf bags of $\X'$, then either it occurs in two different leaf bags of $\Y$ and hence in $\Y(s_{\text{root}}) = \X'(t_{\text{root}})$, or in two leaf bags $\X(t)$ and $\X(t')$ with $s(t)\neq s(t')$ and hence in $\X(t_{\text{root}})$, which means it occurs in the adhesions of $\X(t)$ and $\X(t)'$; these adhesion are contained in $\Y(s(t))$ and $\Y(s(t'))$ respectively, so the vertex must be contained in $\Y(s_{\text{root}}) = \X'(t_{\text{root}})$ as well.
	Finally to check the sizes of adhesions, observe that each adhesion of $(T,\X')$ is either (for leaf nodes $t$ coming from $T$) a subset of an adhesion in $\X$, namely $\X'(t) \cap \X'(t_{\text{root}}) \subseteq \X(t) \cap \X(t_{\text{root}})$,
	or (for leaf nodes $s$ coming from $S$) equal to an adhesion of $(S,\Y)$:
		$$\X'(s) \cap \X'(t_{\text{root}}) = \left(\Y(s) \cup \bigcup_{t \colon s(t) = s} \X(t) \right) \cap \Y(s_{\text{root}}) = \Y(s) \cap \Y(s_{\text{root}}),$$
	since every vertex in $\X(t) \cap \Y(s_{\text{root}}) \subseteq \X(t) \cap \X(t_{\text{root}})$ is contained in $\Y(s(t))$.

	To show that the new decomposition respects $\T$, suppose to the contrary that some leaf bag contains the large side $B$ of a separation $(A,B) \in \T$. Then either this bag is $\X'(t) = \X(t)$ for $t\in V(T)$, contradicting that $(T,\X)$ respects $\T$, or $\X'(s) = \Y(s) \cup \bigcup_{t \colon s(t) = s} \X(t)$ for some $s \in V(S)$.
	Then since $(\X'(s),\bigcup_{s' \neq s \in V(S)} \X'(s'))$ is a separation induced by the adhesion at $\X'(s)$, of order $\leq \alpha'$ (as shown in the previous paragraph), $\X'(s)$ (which contains $B$) must be the large side of this separation in $\T$. 
	But this is an extension of the separation $(\Y(s),\bigcup_{s' \neq s \in V(S)} \Y(s'))$, contradicting~\eqref{eq:respect}.

	The part that is not entirely straightforward is bounding the degrees in the root torso of the new star decomposition $(T',\X')$.
	Observe that this new torso can be obtained from the root torso $\Groot$ of $(T,\X)$ by taking the subgraph induced by $ \X'(t_{\text{root}}) = \Y(s_{\text{root}})$, removing certain virtual edges (from adhesions corresponding to deleted leaves) and adding virtual edges between every two vertices of each new adhesion $\X'(s) \cap \X'(t_{\text{root}})$, $s\in V(S)$.
	We need to bound the number of such edges added to a vertex $v \in \Y(s_{\text{root}})$, assuming it was not already one of the $|V(H)|$ high-degree vertices of the root torso of $(T,\X)$.
	Thus $v$ had at most $d$ neighbors in $\Groot$.
	Since for each new adhesion $\X'(s) \cap \X'(t_{\text{root}}) = \Y(s) \cap \Y(s_{\text{root}})$ (as shown above) which contains $v$ we assumed that $v$ has a neighbor in $\Y(s) \setminus \Y(s_{\text{root}})$ (and these are disjoint sets, for distinct $s$), it follows that $v$ received new edges from at most $d$ distinct new adhesions, at most $\alpha'$ edges from each.
	Therefore, the degree of all but $|V(H)|$ vertices in the root torso of $(T',\X')$ is bounded by $d + \alpha' \cdot d$, that is, Corollary~\ref{cor:groheMarx} holds with this constant in place of $d$.
\end{claimproof}
This concludes the proof of Lemma~\ref{lem:topoLocalStructure}.
\end{proof}

As before, the global decomposition stated in Theorem~\ref{thm:topo-decomposable} follows from the Local Structure Theorem in Corollary~\ref{cor:groheMarx} improved by the bound on bag size in Lemma~\ref{lem:topoLocalStructure}.
The construction only changes in that it can take root bags from the local structure (the star decomposition) given by the first outcome of Corollary~\ref{cor:groheMarx} and use them as bags of the global decomposition, without any modifications.
We note this can also be seen as an instance of a general theorem (11.1) by Robertson and Seymour~\cite{RobertsonS91} which allows to turn `local decompositions' (with respect to a tangle) into global ones; however, giving a precise translation would involve too much notational overhead.

\section{Adding a modulator} \label{sec:modulator}
\newcommand{\modn}{s}

In this section we prove Theorem~\ref{thm:minor-mod-main}
in a more general setting of Section~\ref{sec:turing}.
More precisely, Theorem~\ref{thm:minor-mod-main} follows directly 
from the following theorem via Theorems~\ref{thm:minor-decomposable} and~\ref{thm:topo-decomposable}.

\begin{theorem}\label{thm:mod-main}
One can solve in polynomial time a given \kPath{} instance $(G,k)$,
given access to a set $M \subseteq V(G)$ such that $G-M$ admits a tree decomposition
of width less than $\tdw$ and adhesion $\tdadh = \Oh(1)$, and an oracle
that solves the $\kLinkage{}$ problem for instances $(G',k',S,(R_i)_{i=1}^r)$
with $G'$ being a subgraph of $G$, $r,k' \leq k$, $|S| \leq |M| + \Oh(1)$,
and $|V(G')|$ being bounded
polynomially in $k$, $\tdw$, and $|M|$.
\end{theorem}

Contrary to Section~\ref{sec:turing}, in this section we are not precise with the
polynomial dependency on the parameters $k$, $\tdw$, and $|M|$, as well as the dependency
on the adhesion $\tdadh = \Oh(1)$
of the decomposition. This is due to the fact that the latter dependency on $\tdadh$
is substantial, most notably in the exponent of the polynomial bound on the size
of the oracle calls. Therefore, the result of this section is a purely theoretical result
classifying the aforementioned parameterization as admitting a polynomial Turing kernel.

\paragraph{Decomposition.}
We apply the algorithm of Theorem~\ref{thm:unbreakable} to the graph $G-M$
and separation size $\tdadh$.
Since $\tdadh = \Oh(1)$, the algorithm runs in polynomial time.
By Lemma~\ref{lem:unbreakableLocal}, the output decomposition $(T_0,\X_0)$
has width $\tdw^{\Oh(1)}$ and adhesion $\Oh(1)$.

We modify the decomposition $(T_0,\X_0)$ as follows. For every node $t \in V(T)$
we group the children $t'$ of $t$ according to the adhesions $\X(t) \cap \X(t')$.
For every group $U \subseteq N_T(t)$, we add a new child
$t_U$ of $t$ with $\X(t_U) = \X(t) \cap X(t')$, $t' \in U$, and reattach
the nodes of $U$ as children of $t_U$.
Then, we duplicate $t_U$ into $\max(1, |U|-1)$ copies, arranged in a binary tree, 
  with the nodes of $U$ being different leaves of the binary tree.
Let $(T,\X)$ be the final decomposition.
In this manner, we maintain the property that $(T,\X)$ has width $\tdw^{\Oh(1)}$
and adhesion $\Oh(1)$, while every node has a polynomial in $\tdw$ number of children.

\paragraph{Marking.}
We continue the proof of Theorem~\ref{thm:mod-main} with a somewhat standard marking process.
For an integer $0 \leq k' \leq k-2$ and two vertices $u, v \in M$, define a
\emph{$(u,v,k')$-path} in $G$ as a (simple) path with $k'+2$ vertices: two endpoints $u$ and $v$
and exactly $k'$ internal vertices, all in $G-M$.
Note that a $(u,v,0)$-path is a path consisting of an edge $uv$.
Similarly, for $0 \leq k' \leq k-1$ and a vertex $u \in M$, a $(u,k')$-path
in $G$ is a path with $k'+1$ vertices: one endpoint being $u$ and all other vertices
in $G-M$.

For every $0 \leq k' \leq k-2$ and every $u,v\in M$, we compute an inclusion-wise maximal
family $\mathcal{P}_{u,v,k'}^0$ of internally vertex-disjoint $(u,v,k')$-paths.
Similarly, for every $0 \leq k' \leq k-1$ and $u \in M$
we compute an inclusion-wise maximal family $\mathcal{P}_{u,k'}^0$ of internally
vertex-disjoint $(u,k')$-paths. 
The computation can be done via the previously established kernel for decomposable graph classes.
\begin{lemma}
The families $\mathcal{P}_{u,v,k'}^0$ and $\mathcal{P}_{u,k'}^0$ can be computed
in polynomial time with the access to the assumed \kLinkage{} oracle.
\end{lemma}
\begin{proof}
We focus on $\mathcal{P}_{u,v,k'}^0$; the proof for the second family is analogous.
Since the required family is inclusion-wise maximal, it suffices to find its paths
one-by-one. That is, we show how to find,
in a given induced subgraph of $G-(M \setminus \{u,v\})$ obtained by removing the internal vertices of the previously-found paths, a next path of the desired length $k' \leq k$
between $u$ and $v$ (or conclude that none exists and the packing is maximal). For that, we use a slight modification of the 
previously established kernel in Section~\ref{sec:turing} for decomposable graph classes
(Corollary~\ref{cor:kernel-decomp}).

More precisely, first observe that the Reduction Rule is still correct when we are looking for a $k'$-path between fixed endpoints $u$ and $v$ (as opposed to any $k'$-path) as long as the endpoints do not lie in the set $A$.
Thus, we can (just as in Corollary~\ref{cor:kernel-decomp}) iteratively use the decomposition $(T,\X)$ of $G-M$, reintroduce $u$ and $v$ to every bag, and find a separation $(A',B')$ by Lemma~\ref{lem:decomp2sep} such that the Reduction Rule applies, and iterate this until we find the $k'$-path between $u$ and $v$ or conclude there is none, as in Lemma~\ref{lem:generic}.
%This allows us to decide if the currently held family $\mathcal{P}_{u,v,k'}^0$ is maximal and,
%if not, find a new path to add to it. The lemma follows.
\end{proof}

We define $\mathcal{P}_{u,v,k'}$ to be a subfamily of $k+1$ arbitrarily chosen 
paths from $\mathcal{P}_{u,v,k'}^0$, or let $\mathcal{P}_{u,v,k'} = \mathcal{P}_{u,v,k'}^0$
if $|\mathcal{P}_{u,v,k'}| < k+1$. We define $\mathcal{P}_{u,k'}$ similarly.

Let $|M| = \modn$.
Let $X_1$ be the set of all vertices of $G-M$ that lie on some path
in any of the sets $\mathcal{P}_{u,v,k'}$ or $\mathcal{P}_{u,k'}$ ($u,v\in M$, $k'\leq k$). Note that,
   as we picked at most $k+1$ paths to the families 
   $\mathcal{P}_{u,v,k'}$ and $\mathcal{P}_{u,k'}$, each consisting of at most~$k-1$ vertices from~$G-M$, for each of at most~$k$ choices of the length~$k' \leq k$, and~$\binom{\modn}{2} + \modn \leq 2\modn^2$ choices for the endpoint(s), we have
$$|X_1| < (k+1) \cdot (k-1) \cdot k \cdot 2\modn^2.$$

For every $x \in X_1$, pick one node $t(x) \in V(T)$ such that $x \in \X(t(x))$.
Let $B_1 = \{t(x) \mid x \in X_1\}$.
Let $B_2 \subseteq V(T)$ consist of all of $B_1$, the root of $T$, and the lowest
common ancestor of every pair $(t_1,t_2) \in B_1 \times B_1$.
Let $X_2 = \X(B_2)$.
It is folklore that $|B_2| \leq 2|B_1|$, hence:
\begin{align*}
|B_2| &\leq 2(k+1) \cdot (k-1) \cdot k \cdot 2\modn^2,\\
|X_2| & = \tdw^{\Oh(1)} k^3 \modn^2.
\end{align*}

Let $\CC$ be the set of connected components of $T-B_2$.
As $B_2$ is closed under taking lowest common ancestors of nodes,
we have that $|N_T(C)|\leq 2$ for any~$C \in \CC$: one of the elements of $N_T(C)$ is a parent
of the root of $C$ (if it exists), and additionally one vertex of $C$ may have a child in $N_T(C)$.

\paragraph{Reducing a large component.}
Our main goal is to show that if the bags of any component $C$
contain in total more than $m = \mathrm{poly}(k, \tdw, \modn)$ vertices of $G$, that is,
$|\X(C)| > m$, then 
one can use the Reduction Rule from Section~\ref{sec:turing} to reduce the graph $G$ by at least
one vertex. Fix such a component $C$.

Let $t_0 \in V(C)$ be the lowest node in $C$ such that,
if $D$ is the set of nodes of $C$ in the subtree of $T$ rooted in $t_0$,
then $|\X(D)| > m$.

Let $Z = \X(D) \cap X_2$.
Recall that $|N_T(C)| \leq 2$ due to the fact that $B_2$ is closed under taking lowest common ancestors. 
By the properties of the tree decomposition, every vertex of $Z$ lies in $\X(t)$ for some
$t \in N_T(C)$. Since the adhesions are of constant size, we have that $|Z| = \Oh(1)$.

Let $A = \X(D) \setminus Z$. Note that $A \cap X_2 = \emptyset$.
Since in $(T,\X)$ every node has a number of children bounded polynomially in $\tdw$,
the choice of $t_0$ ensures that 
$|\X(D)| \leq \mathrm{poly}(\tdw) \cdot m$.
Furthermore, note that $N_G(A) \subseteq Z \cup M$; in particular, $|N_G(A)| \leq |M| + \Oh(1)$.

The marking process ensures the following crucial property.
\begin{lemma}\label{lem:guard}
The set $A$ is guarded by $Z$.
\end{lemma}
\begin{proof}
Let $P$ be a $k$-path in $G$ that contains a maximum possible number of vertices
from $X_1$; we claim that $P$ is guarded by $Z$ in $A$.
Assume the contrary, let $Q$ be a traverse of $P$ through $A$ for which neither endpoint
belongs to $Z$. Since $P$ is not completely contained in $G[A]$, and
$N_G(A) \subseteq Z \cup M$, we have that $Q$ is either a $(u,v,k')$-path
for some $u,v \in M$ and $1 \leq k' \leq k-2$ or 
a $(u,k')$-path for some $u \in M$ and $1 \leq k' \leq k-1$. 
Assume the first case; the proof for the second case is analogous.

Since $A$ does not contain any vertex of $X_2 \supseteq X_1$, from the maximality
of $\mathcal{P}_{u,v,k'}^0$ we infer that $\mathcal{P}_{u,v,k'} \subsetneq \mathcal{P}_{u,v,k'}^0$ (since it contains $Q$)
and, consequently, $|\mathcal{P}_{u,v,k'}| = k+1$. Thus, at least one $(u,v,k')$-path
$Q' \in \mathcal{P}_{u,v,k'}$ has no internal vertices on the $k$-path $P$. 
Consequently, by replacing $Q$ with $Q'$ on $P$, we obtain a simple $k$-path $P'$
with strictly more vertices of $X_1$ than $P$. This is a contradiction.
\end{proof}

Thus, if $m = \mathrm{poly}(k, \tdw, \modn)$ is large enough, we can apply the Reduction
Rule for the set $A$ and the guard $Z$; recall that $|Z| = \Oh(1)$
while $m - \Oh(1) < |A| \leq \mathrm{poly}(\tdw) \cdot m$.
The Reduction Rule deletes at least one vertex after invoking
a number of calls to the \kLinkage{} oracle on 
the subgraph $G[N[A]]$; note that $N[A] \subseteq A \cup M \cup Z$.

\paragraph{Wrap up.}
Recall that we have shown that one can apply the Reduction Rule
if there exists a component $C \in \CC$ with
$|\X(C)| > m$ for some $m$ bounded polynomially in $k$, $\tdw$, and $\modn$.
However, since $|B_2| \leq 2(k+1) \cdot (k-1) \cdot k \cdot 2\modn^2$ and a node of $T$ has $\mathrm{poly}(\tdw)$ children,
there are at most $\mathrm{poly}(\tdw)  k^3 \modn^2$ components.
Consequently, if the reduction is not applicable,
 we have $|V(G)|$ bounded polynomially in $k$, $\tdw$, and $\modn$.
Such an instance can be solved with a single call to \kLinkage{} with $k' = k$, $S = \emptyset$,
$r=1$, and $R_1 = \emptyset$.
This finishes the proof of Theorem~\ref{thm:mod-main} and, consequently, of Theorem~\ref{thm:minor-mod-main}.

\section{Conclusions}\label{sec:conc}
We significantly extended the graph classes on which \kPath has a polynomial Turing kernel. In addition, we showed that even an instance that does not belong to such a class, but has a small vertex modulator whose deletion makes it so, can be solved efficiently using small queries to an oracle. A subdivision-based argument (cf.~\cite{GajarskyHOORRVS17}) shows that we cannot generalize much beyond $H$-topological-minor-free graphs without settling the problem in general. In particular, the existence of a polynomial Turing kernel for graphs of bounded expansion implies its existence in general graphs.

While our narrative focused on \kPath, after small modifications our techniques can also be applied to prove analogues of Theorems~\ref{thm:topo-minor-main} and~\ref{thm:minor-mod-main} for the \kCycle problem of detecting a simple cycle of length \emph{at least}~$k$. The main difficulty in adapting our arguments to \kCycle is the fact that, a priori, the only cycles of length at least~$k$ may be arbitrarily much larger than~$k$. However, this issue can easily be resolved in the following way. Since a cycle is contained within a single biconnected component, a Turing kernelization can decompose its input into biconnected components and solve the problem independently in each of them. We then start by testing for the existence of a \emph{path} with~$k^2$ vertices using the algorithms developed in the paper. If there is a path of length~$k^2$ in a biconnected component, then by a classic theorem of Dirac~\cite{Dirac52} there is a cycle of length at least~$k$, and we are done. If no such path exists, then the longest cycle in~$G$ has length less than~$2k$, and we can continue under the guarantee that the cycle we are looking for has length at least~$k$ and less than~$2k$. In this setting, our arguments can be easily adapted. In particular, the absence of a path of length~$k^2$ implies the existence of suitable tree decompositions from which reducible separations can be extracted.

A significant portion of the technical work in this paper was devoted to modifying the graph minors decomposition to obtain the win/win that either answers the problem or finds a reducible separation. In this way, the algorithmic question has driven a challenging graph-theoretic project. It would be interesting to find more problems amenable to such an approach. We conclude with some concrete open problems. Does \kPath have a polynomial Turing kernel on chordal graphs? How about \textsc{Induced $k$-Path} or \textsc{Directed $k$-Path}, on planar graphs?

\bibliography{turing-kernel-minor}
\bibliographystyle{abbrvurl}

\newpage
\appendix
\section{Omitted proofs}

\subsection{Proof of Lemma~\ref{lem:adh-deg-bound}}\label{app:adh-deg-bound}

Consider a node $t \in V(T)$, let $T_t$ be the subtree of $T$ rooted in $t$, and let
$G_t = G[\X(T_t)]$. Note that $(T_t, \X)$ is a tree decomposition of $G_t$ with $t$ being
a root; with regards to this decomposition, we consider the torso $H_t := \torso(G_t, \X(t))$.
Note that the connectivity of $(T,\X)$ implies that 
$H_t$ can be obtained from $G_t[\X(t)]$ by turning the neighborhood $N_{G_t}(C)$ into a clique
for every connected component $C$ of $G_t-\X(t)$. Hence every adhesion to a child of~$t$ corresponds to a clique in~$H_t$, and we may bound the adhesion degree by bounding the number of distinct cliques in~$H_t$.

To show the bound on the number of cliques, we show a longer, more elaborate process
that turns $G_t$ into $H_t$.
Let $\CC$ be the family of connected components of $G_t-\X(t)$.
For every $C \in \CC$
and every vertex $v \in N(C) \subseteq \X(t)$, pick one neighbor $x_{v,C} \in N(v) \cap C$. 
Let $S_C$ be a minimal connected subgraph of $C$ that contains every vertex $x_{v,C}$ for $v \in N(C)$. Note that by minimality, $S_C$ is a tree with at most $|N(C)| \leq h$ leaves and thus
less than $h$ vertices of degree at least $3$. 
Delete from $G_t$ all vertices and edges of $G[C]$ that do not lie in $S_C$, and suppress
every node of degree $2$ in $S_C$ that is not of the form $x_{v,C}$ for $v \in N(C)$.
By performing this operation for every component $C \in \CC$, we obtain
a topological minor $G_t^1$ of $G_t$ with the following property:
the graph $G_t[\X(t)]$ remains untouched in $G_t^1$, while every $C \in \CC$
is turned into its topological minor $C^1$ with $N(C) = N(C^1)$ and $|C^1| < 2|N(C)| \leq 2h$.
Since $G$ does not contain $H$ as a topological minor, neither does $G^1_t$.

We now switch to the theory of graphs of bounded expansion; for an introduction and more notation
we refer
to the textbook~\cite{sparsity}.
A graph class $\GG$ is of \emph{bounded expansion} if there 
exists a function $\nabla: \mathbb{Z}_{\geq 0} \to \mathbb{Z}_{\geq 0}$
such that for every radius $r$, every graph $H \in \GG$,
and every $r$-shallow minor $H'$ of $H$, we have $|E(H')| / |V(H')| \leq \nabla(r)$.

Let $G^2_t$ be a graph obtained from $G^1_t$ by contracting every component $C^1$ for $C \in \CC$
into a single vertex $x_C$. Then, as $|C^1| < 2h$, $G^2_t$ is an $2h$-shallow minor of $G^1_t$.
Let $G^3_t$ be a graph obtained from $G^2_t$ by replacing every vertex $x_C$ for $C \in \CC$ with
a clique $K_C$ of $|N(C)|$ copies of $x_C$. 
Finally, note that $H_t$ can be obtained from $G^3_t$ by contracting every vertex of $K_C$
onto a distinct vertex of $N(C)$, that is, $H_t$ is a $1$-shallow minor of $G^3_t$.

Since taking an $\Oh(1)$-shallow-minor or replacing every vertex with a clique of constant size
in a graph from a class of bounded expansion gives a graph also from a class of bounded expansion
(but possibly worse function $\nabla$)~\cite{sparsity},
  we have that there exists a graph class $\GG$
of bounded expansion (with the function $\nabla$ depending on $H$ and $h$) such that
$H_t \in \GG$. 
Since graphs of bounded expansion have bounded degeneracy and therefore have linearly many cliques~\cite[Lemma 3.1]{sparsity}, we have that
$H_t$ contains at most $f'(h, H) \cdot |\X(t)|$ cliques for some constant $f'(h,H)$
depending on $h$ and $H$.
Thus, in Lemma~\ref{lem:adh-deg-bound} we can take $f(h, H) = 1 + f'(h, H)$ to accommodate
for the additional adhesion corresponding to the parent of~$t$.

\end{document}